%% file: main.tex
\documentclass[sigconf,balance=false]{acmart}
\usepackage{popets}

\setcopyright{popets}
\copyrightyear{YYYY}

\acmYear{YYYY}
\acmVolume{YYYY}
\acmNumber{X}
\acmDOI{XXXXXXX.XXXXXXX}
\acmISBN{}
\acmConference{Proceedings on Privacy Enhancing Technologies}
\input{_pkgs_confs.tex}
\begin{document}

%% The "title" command has an optional parameter,
%% allowing the author to define a "short title" to be used in page headers.
\title{\ours: Trajectory Collection in Continuous Space under Local Differential Privacy}

%%%%%%%%%%%%%%%% Authors' Info %%%%%%%%%%%%%%%%%
%%
%% The "author" command and its associated commands are used to define
%% the authors and their affiliations.

\author{Ye Zheng}
\affiliation{%
  \institution{Rochester Institute of Technology}
  \city{Rochester}
  \country{USA}}

\author{Yidan Hu}
\affiliation{%
  \institution{Rochester Institute of Technology}
  \city{Rochester}
  \country{USA}}

%%
%% By default, the full list of authors will be used in the page
%% headers. Often, this list is too long, and will overlap
%% other information printed in the page headers. This command allows
%% the author to define a more concise list
%% of authors' names for this purpose.

% \renewcommand{\shortauthors}{Zheng et al.}

% As a general rule, do not put math, special symbols or citations
% in the abstract
\begin{abstract}
Trajectory collection is essential for location-based services, yet it can reveal highly sensitive information about users, such as daily routines and activities, raising serious privacy concerns.
Local Differential Privacy (LDP) offers strong privacy guarantees for users even when the data collector is untrusted.
However, existing trajectory collection methods under LDP are largely confined to discrete location spaces, where the size of the location space affects both privacy guarantees and trajectory utility.
Moreover, many real-world applications---such as flying trajectories or wearable-sensor traces---naturally operate in continuous spaces, 
making these discrete-space methods inadequate.

This paper shifts the focus from discrete to continuous spaces for trajectory collection under LDP.
We propose two methods:
\oursd, which perturbs the direction and distance of locations, and \oursc, which perturbs the Cartesian coordinates of locations.
Both methods are theoretically and experimentally analyzed for trajectory utility in continuous spaces.
\ours can also be applied to discrete spaces by rounding perturbed locations to any discrete space embedded in the continuous space.
In this case, the privacy and utility guarantees of \ours are independent of the number of locations in the space,
and each perturbation requires only $\Theta(1)$ time complexity.
Evaluation results on discrete location spaces validate the efficiency advantage and
demonstrate that \ours outperforms state-of-the-art methods with improved
trajectory utility, particularly for large privacy parameters.
\end{abstract}

%%
%% Keywords. The author(s) should pick words that accurately describe
%% the work being presented. Separate the keywords with commas.

% make the title area
\maketitle

\input{introduction.tex}
\input{preliminaries.tex}

\input{our_method_1.tex}
\input{our_method_2.tex}
\input{experiments.tex}

% The authors would like to thank...

% \clearpage
\bibliographystyle{ACM-Reference-Format}
\bibliography{reference}
% \clearpage
\appendix
\input{related_work.tex}
\input{appendix.tex}

\end{document}

%% file: _pkgs_confs.tex
\usepackage[T1]{fontenc}
\usepackage{microtype} % micro-typography for tighter typesetting

\usepackage[table]{xcolor}

\usepackage[symbol]{footmisc} % symbol footnote

\usepackage{amsmath,amsthm}
\usepackage{mathtools} % for \coloneqq etc.

\DeclareMathOperator*{\argmin}{arg\,min}

\usepackage{soul,bm} 
\usepackage{fancyvrb} % for verbatim environment
\usepackage[most]{tcolorbox} % for highlighted box

\usepackage[font=small]{caption} % caption format
\usepackage{subcaption,wrapfig} % sub-figure and wrap-figure
\usepackage{graphicx}

\usepackage{tabularx}
\usepackage{booktabs} % for \toprule, \midrule, \bottomrule
\usepackage{threeparttable} % for table notes
\usepackage{multirow} % for multi-row cells
\usepackage{makecell} % for line break in table cells
\usepackage{pifont}   % for check and cross symbols etc.

\usepackage{algorithmic}
\usepackage[ruled,vlined,linesnumbered]{algorithm2e} 

\usepackage{tikz}
\usetikzlibrary{fit,calc}
\newcommand*{\tikzmk}[1]{\tikz[remember picture,overlay,] \node (#1) {};\ignorespaces}
\colorlet{mypink}{magenta}
\colorlet{myblue}{cyan}
\newcommand{\boxit}[1]{\tikz[remember picture,overlay]{\node[yshift=1.5pt,fill=#1,opacity=.15,fit={(A)($(B)+(.88\linewidth,.99\baselineskip)$)}] {};}\ignorespaces}
\newcommand{\boxitcd}[1]{\tikz[remember picture,overlay]{\node[yshift=3pt,fill=#1,opacity=.15,fit={(C)($(D)+(.88\linewidth,.9\baselineskip)$)}] {};}\ignorespaces}
\newtheorem{definition}{\textbf{Definition}}
\newtheorem{theorem}{\textbf{Theorem}}

\newtheorem{example}{\textbf{Example}}

\newcommand{\ours}{TraCS\xspace}
\newcommand{\oursd}{TraCS-D\xspace}
\newcommand{\oursc}{TraCS-C\xspace}

\newcommand{\traj}{\mathcal{T}}
\newcommand{\mechanism}{\mathcal{M}}
\newcommand{\traspace}{\mathcal{S}}

\newcommand{\pdomain}{\mathcal{D}_\varphi}

%% file: introduction.tex
\section{Introduction} \label{sec:introduction}
Trajectory data from users---sequences of locations describing movement over time---are a fundamental resource for activity analysis and location-based services, 
including activity classification and routine detection.
However, directly collecting trajectory data raises significant privacy concerns,
especially given users' growing privacy expectations and stringent regulations~\cite{gdpr,DBLP:journals/puc/Krumm09}.
Anonymization techniques like $k$-anonymity~\cite{DBLP:journals/ijufks/Sweene02,DBLP:conf/icde/LiLV07} have proven insufficient 
to protect trajectory privacy~\cite{DBLP:journals/corr/abs-1112-2020}.
Thus, stronger privacy guarantees are necessary for such data.

Local differential privacy (LDP) provides a provable worst-case privacy guarantee 
by ensuring the hardness of distinguishing a user's data from any released data.
The privacy level of an LDP mechanism is quantified by its privacy parameter $\varepsilon$,
which also generally implies the extent of data perturbation.
Through LDP, users can release their perturbed trajectory data to an untrusted data collector
while still benefiting from the services provided by the collector.

Existing LDP methods for trajectory collection are designed for discrete spaces.
They leverage discrete LDP mechanisms, such as the Exponential mechanism~\cite{DBLP:conf/innovations/NissimST12},
to perturb the trajectory data.
Discrete LDP mechanisms are explicitly defined on the size of the location space. 
As a result, these methods either rely on properly splitting the entire location space into grids~\cite{DBLP:conf/ccs/WangH0QH22} 
or assume that the location space consists of a finite set of labeled locations (i.e. points of interest)~\cite{DBLP:journals/pvldb/CunninghamCFS21,DBLP:journals/pvldb/Zhang000H23}. 
\ However, relying on discrete location spaces has three inherent limitations:
(i) Their privacy guarantees are inherently tied to the discrete set.
For example, a space with $10$ locations offers weaker effective privacy than one with $100$ locations:
even a trivial inference strategy that always outputs a fixed location succeeds with probability at least $1/10$ in the former,
regardless of the privacy parameter $\varepsilon$.\footnote{
    Appendix~\ref{appendix:space_dependence} provides details on the space-dependence of indistinguishability.}
(ii) Their efficacy and efficiency are often limited by the domain size.
As the number of candidate locations grows, the probability of a mechanism outputting the true location diminishes.
Furthermore, the widely used Exponential mechanism incurs linear sampling complexity in the domain size,
making the generation of each perturbed location computationally expensive.\footnote{
    Appendix~\ref{appendix:exponential_mechanism} details the Exponential mechanism's efficiency and data utility (efficacy) limitations.}
(iii) Discrete methods are not directly applicable to inherently continuous location spaces,
such as those arising from flying and sailing trajectories or sensor data from wearable devices.
Although one can discretize a continuous space before applying discrete methods,
this inherits the aforementioned limitations. Moreover, choosing a suitable discretization granularity that balances privacy, utility, and efficiency is non-trivial.\footnote{
    Appendix~\ref{appendix:discrete_to_continuous} further discusses the challenges of adapting discrete mechanisms to continuous spaces.}

To address these limitations, this paper shifts the focus from discrete to continuous spaces for trajectory collection under LDP.
A continuous space represents locations as real-valued points,
such as GPS coordinates in $[-180, 180] \times [-90, 90]$,
and thus contains infinitely many candidate locations.
Collecting trajectory data directly in continuous spaces is natural for many applications
and offers three key advantages.
(i) The privacy guarantee does not depend on the cardinality of a discretized location set
(ii) The sampling mechanism operates directly on the continuous space, so its efficacy and efficiency are not constrained by the space size.
(iii) Perturbed locations can be post-processed (e.g. rounded) to any discrete space embedded within the continuous space,
making the approach applicable to both continuous and discrete settings.

Designing LDP trajectory collection methods for continuous spaces poses unique challenges:
they must guarantee LDP for every location in the continuous domain;
they must preserve sufficient trajectory utility for downstream tasks relative to the original trajectories (formalized in Section~\ref{sec:evaluation});
and they must generate perturbed locations efficiently to support real-time applications.

This paper proposes two LDP methods for trajectory collection in continuous spaces.
The main insight is to decompose the $2$D continuous space into $1$D subspaces and design mechanisms for each subspace.
Our mechanisms build on existing utility-optimized piecewise-based mechanisms~\cite{DBLP:journals/popets/ZhengMH25} 
for $1$D bounded numerical domains.\footnote{
    Another category of LDP mechanisms for bounded numerical domains is truncated mechanisms (e.g. the truncated Laplace mechanism~\cite{DBLP:journals/tdp/000619, DBLP:journals/jpc/HolohanABA20}).
    While such mechanisms can be incorporated into \ours, they are more complex and typically less effective than piecewise-based mechanisms.
    Section~\ref{subsubsec:other_bounded_ldp} provides details, and Section~\ref{sec:evaluation} includes experimental comparisons.
    }
\ Based on two different decompositions of the continuous space, we propose two methods: \oursd and \oursc.
The method \oursd treats the continuous space as the composition of a direction space and a distance space.
We exploit the direction information of trajectories and
design a direction perturbation mechanism to ensure LDP for the direction space.
Combined with a piecewise-based mechanism for the distance space,
\oursd ensures LDP for the continuous $2$D space.
The method \oursc treats the continuous space as the Cartesian product of two distance spaces.
Each location is represented by a $2$D Cartesian coordinate, and \oursc perturbs these coordinates.
\ Both methods guarantee LDP for the entire $2$D continuous space. 
They can also be applied to any discrete space embedded in the continuous domain
by rounding perturbed locations to their nearest discrete points.
Compared with existing methods based on the Exponential mechanism~\cite{DBLP:journals/pvldb/CunninghamCFS21,DBLP:journals/pvldb/Zhang000H23},
\ours is unaffected by the number of locations in the location space and require only $\Theta(1)$ time to generate each perturbed location.
Specifically, our contributions are as follows:
\begin{itemize}
    \item To the best of our knowledge, this is the first work to develop trajectory collection methods
        for continuous spaces under pure LDP.
        We highlight the benefits of operating directly in continuous spaces (rather than discretizing) and propose \oursd and \oursc accordingly.
        \ Our key insight is to decompose the $2$D continuous space into two $1$D subspaces and to build $2$D trajectory perturbation mechanisms
        from existing utility-optimized $1$D piecewise-based mechanisms, leveraging the direction and coordinate information in continuous trajectories.
        We theoretically and experimentally analyze their trajectory utility.    
\item Our method \ours also applies to discrete spaces.
    Compared with existing approaches for discrete spaces, \ours has substantially lower computational complexity for generating perturbed locations.
    Experiments on discrete location spaces validate this efficiency, with an average runtime below $1\%$ of that of NGram, L-SRR, and ATP~\cite{DBLP:journals/pvldb/CunninghamCFS21, DBLP:journals/pvldb/Zhang000H23, DBLP:conf/ccs/WangH0QH22},
    three state-of-the-art LDP methods designed for discrete spaces.
    In addition, the results show that \ours generally achieves better trajectory utility, particularly at larger privacy parameters.
\end{itemize}

% \textbf{Structure.} The rest of the paper is organized as follows.
% Section~\ref{sec:preliminaries} introduces the preliminaries.
% Section~\ref{sec:method} presents \oursd and \oursc methods, their theoretical analysis,
% along with discussions on their differences and extensions.
% Following that, Section~\ref{sec:evaluation} presents the experimental evaluations on continuous and discrete spaces.

%% file: preliminaries.tex
\section{Preliminaries} \label{sec:preliminaries}
This section formulates the problem, introduces LDP, and reviews three trajectory collection methods under LDP. 
We highlight their limitations in continuous spaces, which motivates the need for new LDP methods for trajectory collection in such settings.

\subsection{Problem Formulation}

A typical trajectory collection schema consists of a set of trajectories from users and one collector.
Each trajectory $\mathcal{T}$ is a sequence of locations $\mathcal{T} = \{\tau_1, \tau_2, \dots, \tau_n\}$, 
where $\tau_i \in \mathcal{S}$ and $\mathcal{S}\subset \mathbb{R}^2$ is a continuous and bounded domain.
The collector needs to collect these trajectories to provide analysis or location-based services.

However, the collector is untrusted and may act as an adversary attempting to infer users' sensitive data.
Therefore, releasing the sensitive trajectories to the collector poses a privacy risk.
To protect privacy, users perturb their trajectories using a privacy mechanism $\mechanism: \mathcal{S} \to Range(\mechanism)$,
and then send the perturbed trajectories $\mathcal{T}' = \{\tau_1', \tau_2', \dots, \tau_n'\}$
to the collector.

We aim to design a mechanism $\mathcal{M}$ that satisfies $\varepsilon$-LDP, providing provable privacy guarantee for sensitive trajectories.
Additionally, $\mathcal{M}$ should preserve utility, 
ensuring perturbed trajectories remain comparable to the originals, while incurring only modest computational overhead suitable for real-time applications.

\subsection{Local Differential Privacy}

\begin{definition}[$\varepsilon$-LDP~\cite{DBLP:conf/focs/DuchiJW13}] \label{def:edldp}
    A perturbation mechanism $\mathcal{M}: \mathcal{X} \to Range(\mechanism)$ satisfies $\varepsilon$-LDP,
    if for two arbitrary input $x_1$ and $x_2$, the probability ratio of outputting the 
    same $y$ is bounded:
    $$
    \forall x_1, x_2 \in \mathcal{X}, \forall y \in Range(\mechanism): \frac{\Pr[\mathcal{M}(x_1) = y]}{\Pr [\mathcal{M}(x_2) = y]} \leq \exp(\varepsilon).
    $$
\end{definition} 

If $\mechanism(x)$ is continuous, the probability $\Pr[\cdot]$ is replaced by probability density function ($pdf$).
Intuitively, Definition~\ref{def:edldp} represents the difficulty of distinguishing
$x_1$ and $x_2$ given $y$.
Lower values of the privacy parameter $\varepsilon \in [0, +\infty)$ mean higher privacy. 
For example, $\varepsilon = 0$ requires $\mathcal{M}$ to map two arbitrary 
inputs to any output $y$ with the same probability, 
thus the perturbed output contains no distribution information of the sensitive input, 
making any hypothesis-testing method to infer the sensitive $x$ powerless.

\begin{theorem}[Sequential Composition of LDP~\cite{DBLP:journals/fttcs/DworkR14,DBLP:journals/pvldb/Zhang000H23}] \label{thm:sequential}
    Let $\mechanism_1$ and $\mechanism_2$ be two mechanisms that satisfy $\varepsilon_1$ and $\varepsilon_2$-LDP, respectively.
    Their composition, defined as $\mechanism_{1,2}\coloneqq (\mechanism_1, \mechanism_2):(\mathcal{X}_1, \mathcal{X}_2)\to (Range(\mechanism_1), Range(\mechanism_2))$, 
    satisfies $(\varepsilon_1 + \varepsilon_2)$-LDP.
\end{theorem}

Sequential composition enables (i) designing LDP mechanisms for 2D data by composing two 1D mechanisms, 
and (ii) perturbing multiple locations in a trajectory (i.e. repeated composition) while still providing an overall LDP guarantee.

\subsubsection{Piecewise-based Mechanism}

When the input domain $\mathcal{X}$ is continuous and bounded, state-of-the-art LDP mechanisms are often \emph{piecewise-based}~\cite{DBLP:conf/icde/WangXYZHSS019,DBLP:conf/sigmod/Li0LLS20,DBLP:journals/popets/ZhengMH25}.
They generate an output by sampling from a probability distribution over $\mathcal{X}$ whose density is defined in a piecewise manner.
Below, we present a unified formulation that captures this family of mechanisms.

\begin{definition}[Piecewise-based mechanism]\label{def:piecewise}
    A piecewise-based mechanism $\mechanism:\mathcal{X}\to Range(\mechanism)$ is a family of probability distributions that,
    given input $x\in\mathcal{X}$, outputs $y\in Range(\mechanism)$ according to 
    \begin{equation*}
        pdf[\mechanism(x)=y] = 
        \begin{dcases}
            p_{\varepsilon} & \text{if} \ y \in [l_{x,\varepsilon}, r_{x,\varepsilon}), \\ 
            p_{\varepsilon} / \exp{(\varepsilon)} & \text{otherwise},
        \end{dcases}
    \end{equation*}
    where $p_{\varepsilon}$ is the sampling probability w.r.t. $\varepsilon$,
    $[l_{x,\varepsilon}, r_{x,\varepsilon})$ is the sampling interval w.r.t. $x$ and $\varepsilon$.
    $Range(\mechanism)\supseteq \mathcal{X}$ is also a continuous and bounded domain.
\end{definition}
Piecewise-based mechanisms sample an output $y$ with a high probability $p_\varepsilon$ from the interval $[l_{x,\varepsilon}, r_{x,\varepsilon})$,
and with a lower probability from the remaining two pieces, while still satisfying $\varepsilon$-LDP.
Representative instantiations include OGPM~\cite{DBLP:journals/popets/ZhengMH25} and PM~\cite{DBLP:conf/icde/WangXYZHSS019} for mean estimation,
and SW~\cite{DBLP:conf/sigmod/Li0LLS20} for distribution estimation.
Our perturbation mechanisms build on the 1D piecewise-based mechanism in OGPM to design new 2D mechanisms tailored to trajectory collection.

\subsubsection{$k$-RR and Exponential Mechanism}

If $\mathcal{X}$ is discrete, especially when $|\mathcal{X}|$ is not large~\cite{DBLP:conf/uss/WangBLJ17}, 
a classical LDP mechanism is $k$-RR (or GRR in some literature)~\cite{DBLP:conf/nips/KairouzOV14}.
\begin{definition} \label{def:krr}
    $k$-RR is a sampling mechanism $\mechanism:\mathcal{X}\to \mathcal{X}$ that, 
    given $|\mathcal{X}|=k$ and input $x\in\mathcal{X}$, outputs $y\in\mathcal{X}$ according to
    \begin{equation*}
        \Pr[\mechanism(x)=y] = 
        \begin{dcases}
            \frac{\exp(\varepsilon)}{|\mathcal{X}|-1+\exp(\varepsilon)} & \text{if} \ y = x, \\ 
            \frac{1}{|\mathcal{X}|-1+\exp(\varepsilon)} & \text{otherwise}.
        \end{dcases}
    \end{equation*}
\end{definition}

$k$-RR outputs the truth $x$ with a higher probability or outputs other values with a lower probability,
while satisfying $\varepsilon$-LDP.

When there is a semantic distance (score) function between elements in $\mathcal{X}$,
the Exponential mechanism~\cite{DBLP:conf/innovations/NissimST12} is more widely used.
Unlike $k$-RR, which treats other values (except $x$) equally, the Exponential mechanism 
leverages a score function to assign different probabilities.
$k$-RR is a special case of the Exponential mechanism when the score function is a binary indicator function.

\begin{definition}[Exponential Mechanism~\cite{DBLP:conf/innovations/NissimST12}] \label{def:exponential_mechanism}
    Given a score function $d: \mathcal{X} \times \mathcal{Y} \to \mathbb{R}$, a privacy parameter $\varepsilon > 0$,
    and a set of possible outputs $\mathcal{Y}$, the Exponential mechanism $\mechanism: \mathcal{X} \to \mathcal{Y}$
    is defined by:
    \begin{equation*}
        \Pr[\mechanism(x)=y] = \frac{\exp \left (\frac{\varepsilon d(x,y)}{2\Delta d}\right )}{\sum_{y' \in \mathcal{Y}} \exp \left (\frac{\varepsilon d(x,y')}{2\Delta d}\right )},
    \end{equation*}
    where $\Delta d = \max_{x,y,y' \in \mathcal{Y}} |d(x,y) - d(x,y')|$ is the sensitivity of the score function $d$.
\end{definition}

In the context of trajectory collection, common score functions include the negative Euclidean distance $d(x,y) \coloneq -||x-y||_2$,
great circle distance~\cite{DBLP:journals/pvldb/Zhang000H23}, etc.
Under such score functions (e.g. $d(x,y) = -\|x-y\|_2$), locations closer to the true location $x$ have larger scores $d(x,y)$ and thus receive higher sampling probabilities, making them more likely to be selected.

\subsubsection{Privacy Parameter for Trajectory Data}

When applying LDP to trajectory data, there are two strategies for setting the privacy parameter.
(i) Treating each location $\tau_i$ in a trajectory $\traj$ as a data item and applying an $\varepsilon$-LDP mechanism to perturb each location;
(ii) Treating the entire trajectory $\traj$ as a data item and applying an $\varepsilon$-LDP mechanism to perturb the whole trajectory.\footnote{
    This concerns how to set the privacy parameter for trajectory data, rather than defining the ``trajectory space'' as the input space in LDP mechanisms.
        We focus on LDP mechanisms for location spaces in this paper, 
        as the trajectory space in $\mathbb{R}^2$ is infinite-dimensional.
}
The first strategy is cleaner and convertible to the second strategy via the Sequential Composition Theorem~\ref{thm:sequential},
whereas the second depends on the trajectory length,
which complicates the analysis and the design of mechanisms under a fixed $\varepsilon$.
Therefore, we primarily adopt the first strategy in this paper, and include experiments for the second in Appendix~\ref{appendix:trajectory_level_epsilon} for completeness.

\subsection{Existing Methods} \label{sec:existing}

\begin{table}[t]
    \begin{center}
        \caption{Comparison with existing methods.}\label{tab:existing}
        \resizebox*{\linewidth}{!}{
            \begin{threeparttable}
                \begin{tabular}{llr}
                    \toprule
                    & \textbf{Main technique} & \textbf{LDP guarantee} \\
                    \midrule
                    NGram~\cite{DBLP:journals/pvldb/CunninghamCFS21} & Hierarchical decomposition  & \multirow{3}{*}[-0.2em]{Discrete space} \\[0.2em]
                    L-SRR~\cite{DBLP:conf/ccs/WangH0QH22} & Staircase RR mechanism & \\[0.2em]
                    ATP~\cite{DBLP:journals/pvldb/Zhang000H23} & Direction perturbation & \\
                    \cmidrule(lr){1-3}
                    This paper (\ours) & \makecell[l]{Direction\tnote{a}~\enspace\& \\ coordinate perturbation} & Continuous space\\
                    \bottomrule
                \end{tabular}
                \begin{tablenotes}
                    \footnotesize \item[a] Designed for continuous direction space (different from ATP).
                \end{tablenotes}
            \end{threeparttable}
        }
    \end{center}
\end{table}

Table~\ref{tab:existing} summarizes three state-of-the-art trajectory collection methods under pure LDP, along with our approach.
Other works that emphasize external knowledge rather than new perturbation mechanisms are discussed in Appendix~\ref{sec:related_work}.

NGram~\cite{DBLP:journals/pvldb/CunninghamCFS21} hierarchically decomposes the physical space $\mathcal{S}$ into fine-grained discrete regions with semantic labels.
Concretely, $\mathcal{S}$ is first partitioned into small spatial regions $R_s$, then further refined into category-specific regions $R_c$, and finally augmented with temporal information $R_t$.
A resulting ``gram'' can be written as $r_{sct}=\{\text{mountain, church, 3am}\}$.
This decomposition represents a trajectory $\traj$ at the knowledge level, enabling the use of public knowledge to define \emph{reachability}, i.e. to constrain the set of feasible perturbed trajectories.
For instance, visiting a church at $3$am is unlikely because it is typically closed, which reduces the reachable set of $\traj'$.
Finally, NGram applies the Exponential mechanism to sample a perturbed trajectory from this reachable set.

% LDPTrace~\cite{DBLP:journals/pvldb/DuHZFCZG23} 
% synthesizes trajectories based on a generative model.
% Specifically, the physical space $\mathcal{S}$ is discretized into cells,
% and each sensitive trajectory is represented as a sequence of cells $\traj =\{c_1, c_2, \dots, c_n\}$.
% Neighboring cells, i.e. $c_i$ and $c_{i+1}$, form transition patterns,
% which are perturbed by LDP mechanisms.
% These perturbed transition patterns are then used to train a Markov chain model that captures the transition patterns between cells.
% Finally, new trajectories $\traj'$ are synthesized by the trained model. 

L-SRR~\cite{DBLP:conf/ccs/WangH0QH22} introduces the staircase randomized response (SRR) mechanism for location perturbation.
Unlike $k$-RR, which assigns only two probabilities---one for the true location and another for all other locations---the SRR mechanism employs a staircase-shaped probability distribution.
L-SRR recursively partitions the physical space $\mathcal{S}$ and forms location groups $G_1, G_2, \dots, G_g$ based on their distances to the true location.
By integrating the SRR mechanism, L-SRR assigns higher probabilities to locations in groups that are closer to the true location, and lower probabilities to those farther away.
This approach is similar to the Exponential mechanism, but uses group-level distance instead of location-level distance.

ATP~\cite{DBLP:journals/pvldb/Zhang000H23} perturbs directions to restrict the set of feasible perturbed locations.
For a trajectory $\traj=\{\tau_1,\tau_2,\dots,\tau_n\}$, when perturbing $\tau_{i+1}$, ATP first perturbs the direction from $\tau_i$ to $\tau_{i+1}$.
Specifically, it partitions the direction space into $k$ sectors and applies $k$-RR to sample a sector.
The perturbed location $\tau_{i+1}'$ is then constrained to lie within the sampled sector, which reduces the candidate set of $\tau_{i+1}'$;
ATP finally samples $\tau_{i+1}'$ from this constrained set using the Exponential mechanism.

We summarize their limitations in the context of continuous spaces as follows:
\begin{itemize}
    \item Existing methods are designed for discrete location spaces. 
    Discretizing a continuous domain is possible, but choosing an appropriate granularity is non-trivial.
    \item Beyond discretization, methods that rely on the Exponential mechanism (e.g. NGram and ATP) typically incur high computational cost for evaluating scores and sampling, 
    and their utility can be sensitive to the spatial distribution of candidate locations. 
    L-SRR requires constructing distance-based groups for every location, leading to similar scalability issues. 
    Moreover, approaches that depend on public knowledge (e.g. NGram) are hard to design in continuous spaces (where the location set is infinite and semantics are difficult to define), 
    are often dataset-specific, and the public knowledge is typically \emph{known} to the adversary, which can further weaken privacy protection.
\end{itemize}

Appendix~\ref{appendix:exponential_mechanism} details the limitations of the Exponential mechanism.
In brief, perturbing a location requires evaluating the score function $d(x,y)$ for all candidate outputs $y$ associated with $x$,
which takes $\Theta(m)$ time where $m$ is the number of locations in the discrete space.
Moreover, generating each perturbed location requires sampling from an $m$-piece cumulative distribution function (CDF),
which also incurs $\Theta(m)$ time.
Such a per-location cost is prohibitive for large-scale location spaces.
Finally, because the sampling probabilities are induced by the score function,
they can be sensitive to the spatial distribution of candidate locations and their relative distances,
which may further degrade utility.

To address these limitations, we propose \ours, a novel trajectory collection method that ensures rigorous LDP guarantees for trajectories in continuous spaces. 
\ours is carefully designed to tackle the privacy, utility, and efficiency challenges inherent in handling continuous domains without relying on discretization.

%% file: our_method_1.tex
\section{Our Method: \ours} \label{sec:method}

This section presents our methods, \oursd and \oursc, along with their theoretical analyses and extensions.
Notations used in this section can be found in Table~\ref{tab:notations}.

\subsection{Location Space} \label{subsec:location_space}

We consider the location space constrained by a pair of longitude and latitude lines,
i.e. $\mathcal{S}\subset \mathbb{R}^2$ is a rectangular area $[a_{\text{sta}},a_{\text{end}}]\times[b_{\text{sta}},b_{\text{end}}]$,
where $a_{\text{sta}},a_{\text{end}}$ denote the longitudes and $b_{\text{sta}},b_{\text{end}}$ denote the latitudes.
Consequently, each location $\tau_i\in \traspace$ has a natural representation as a pair of coordinates $(a_i, b_i)$.
\ This representation aligns with real-world location data (e.g. GPS data) and
existing trajectory collection methods for discrete spaces~\cite{DBLP:journals/pvldb/CunninghamCFS21,DBLP:journals/pvldb/DuHZFCZG23,DBLP:journals/pvldb/Zhang000H23,DBLP:journals/csur/JiangLZZXI21,DBLP:conf/ccs/WangH0QH22}.
% Additionally, the entire Earth is also commonly represented as a bounded continuous space $[-90,90] \times[-180,180]$,
% which is naturally consistent with our location space $\mathcal{S}$.

The core idea behind having various LDP methods for the location space $\mathcal{S}$ stems from its multiple decompositions.
One decomposition means $\mathcal{S}$ can be represented by several subspaces.
For example, in addition to the aforementioned longitude-latitude coordinate representation, each location $\tau_i\in \mathcal{S}$
can also be represented by a direction component and a distance component.
Based on these decompositions, we propose two LDP methods:
\oursd, which perturbs the direction and distance subspaces, 
and \oursc, which perturbs subspaces of the Cartesian coordinates.

\subsection{\oursd}

Continuous space $\mathcal{S}$ can be represented as
the composition of a direction space $\pdomain$ and a distance space $\mathcal{D}_{r(\varphi)}$
at each (reference) location $\tau_i$.
Figure~\ref{fig:decomposition_d} illustrates this decomposition.

\begin{figure}[t]
    \centering
    \begin{subfigure}[b]{0.43\linewidth}
        \centering
        \includegraphics[width=0.98\textwidth]{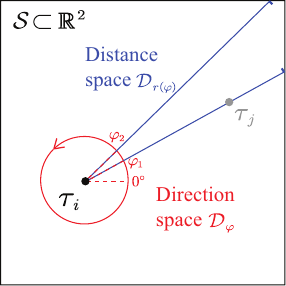}
        \caption{\oursd: $\traspace = \mathcal{D}_{\varphi} \otimes \mathcal{D}_{r(\varphi)}$}
        \label{fig:decomposition_d}
    \end{subfigure}
    \hfill
    \begin{subfigure}[b]{0.43\linewidth}
        \centering
        \includegraphics[width=0.98\textwidth]{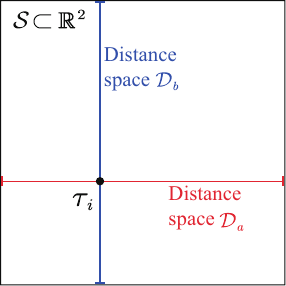}
        \caption{\oursc: $\traspace = \mathcal{D}_{a} \times \mathcal{D}_{b}$}
        \label{fig:decomposition_c}
    \end{subfigure}
    \caption{Two decompositions of location space $\traspace$ at $\tau_i$.
    Any other location $\tau_j \in \traspace$ can be represented by direction-distance coordinates
    $(\varphi, r(\varphi))$ or Cartesian coordinates $(a, b)$.
    }
    \label{fig:decomposition}
\end{figure}

\begin{table}[t]
    \begin{center}
        \caption{Notations used in this paper.} \label{tab:notations}
            \small
                \begin{tabular}{p{3.8cm}l}
                    \toprule
                    \textbf{Notation} & \textbf{Meaning} \\
                    \midrule
                    $\traj \coloneqq \{\tau_1, \tau_2, \ldots, \tau_n\}$ & Sensitive trajectory \\
                    $\traj' \coloneqq \{\tau_1', \tau_2', \ldots, \tau_n'\}$ & Perturbed trajectory \\
                    $\traspace \coloneqq [a_{\text{sta}},a_{\text{end}}]\times[b_{\text{sta}},b_{\text{end}}]$ & Continuous location space \\
                    \midrule
                    $\mechanism_\circ$ & Direction perturbation mechanism \\
                    $\mechanism_{-}$ & Distance perturbation mechanism \\
                    $\varphi \in \pdomain \coloneqq [0, 2\pi)$ & Direction \& direction space \\
                    $\varphi'$ & Perturbed direction \\
                    $r(\varphi) \in \mathcal{D}_{r(\varphi)}$ & Distance \& distance space along $\varphi$ \\
                    $r'(\varphi)$ & Perturbed distance along $\varphi$ \\
                    $\overline{r}(\varphi) \in [0,1)$ & Normalized distance over $\mathcal{D}_{r(\varphi)}$ \\
                    \bottomrule
                \end{tabular}
    \end{center}
\end{table}

\begin{figure*}[th!]
    \centering
    \begin{subfigure}[b]{0.45\textwidth}
        \centering
        \includegraphics[width=0.99\textwidth]{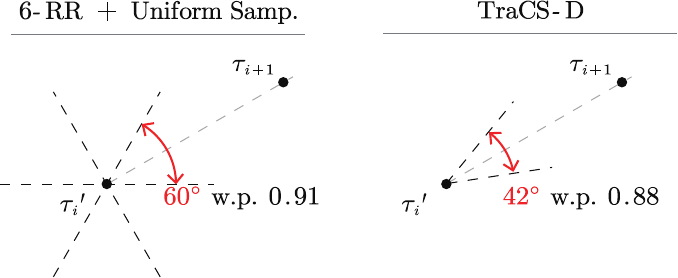}
        \caption{$\varepsilon=4$}
        \label{fig:direction_perturbation_sub1}
    \end{subfigure}
    \hfill
    \begin{subfigure}[b]{0.45\textwidth}
        \centering
        \includegraphics[width=0.99\textwidth]{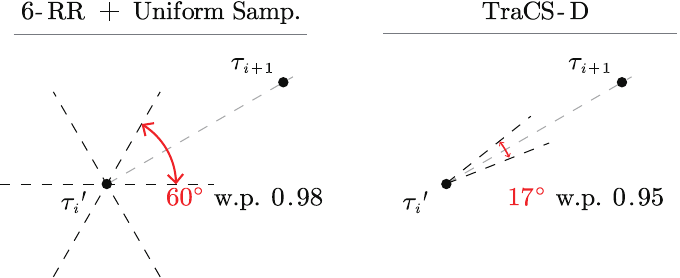}
        \caption{$\varepsilon=6$}
        \label{fig:direction_perturbation_sub2}
    \end{subfigure}
    \caption{Comparison of the dominant sectors of the strawman approach and \oursd with $\varepsilon=4$ and $\varepsilon=6$. 
    The red angular arcs indicate the dominant sectors, with their probabilities shown on the right.
    The dominant sector of \oursd narrows as $\varepsilon$ increases, leading to a smaller inner-sector (sampling) error.
    In contrast, the strawman approach has a fixed dominant sector,
    which is independent of $\varepsilon$ and leads to a large inner-sector error.
    }
    \label{fig:direction_perturbation}
\end{figure*}

\textbf{Direction space $\pdomain$.}
In the space $\mathcal{S}\subset$ $\mathbb{R}^2$, the direction is represented by the angle $\varphi$ relative to a reference direction.
We define the reference direction as the latitude lines, i.e. the $0^{\circ}$ direction. 
Thus, the direction space $\pdomain$ is a circular domain $[0, 2\pi)$ for any location.
Fixing a location $\tau_i$, any other location $\tau_j \neq \tau_i$ has a unique direction with respect to $\tau_i$.

\textbf{Distance space $\mathcal{D}_{r(\varphi)}$.}
The other subspace, i.e. the distance space $\mathcal{D}_{r(\varphi)}$, is determined by the direction $\varphi$.
Its size is the distance from $\tau_i$ to the boundary of $\traspace$ in the direction $\varphi$,
which is a function of $\varphi$.
Since $\mathcal{S}$ is rectangular, $r(\varphi)$ has a closed-form expression, which is detailed later.
Thus, fixing $\tau_i$ and a direction $\varphi$, any other location $\tau_j$ has a unique distance $r(\varphi)$ with respect to $\tau_i$.

By this decomposition, we can represent any location $\tau_j \in S$ as coordinates
$(\varphi, r(\varphi))\in (\mathcal{D}_{\varphi},\mathcal{D}_{r(\varphi)})$ relative to $\tau_i$.
As an example, $\tau_{j}$ in Figure~\ref{fig:decomposition_d} is represented by $(\varphi_1, r(\varphi_1))$.
This representation is unique and reversible:
any pair of coordinates $(\varphi, r(\varphi))$ can be mapped back to a unique location $\tau_j \in \mathcal{S}$.
Thus, if $\tau_j$ is a sensitive location,
an LDP mechanism applied to $\tau_j = (\varphi, r(\varphi))$ is essentially applied to $\varphi$ and $r(\varphi)$.
By Composition Theorem~\ref{thm:sequential}, we can design perturbation mechanisms for $\varphi$ and $r(\varphi)$, respectively, to achieve LDP for $\tau_j$.

\subsubsection{Direction Perturbation} \label{subsubsec:direction_perturbation}

The technical challenge in direction perturbation is designing an LDP mechanism for the circular domain $\pdomain = [0,2\pi)$.
This is non-trivial due to the following properties of the circular domain:
\begin{itemize}
    \item The circular domain $[0,2\pi)$ is bounded but not finite, making the unbounded mechanisms
    and discrete mechanisms inapplicable.
    \item The circular domain has a unique distance metric, e.g. the distance between $0$ and $2\pi$ is $0$, not $2\pi$,
    making the Euclidean distance-based mechanisms inapplicable.
\end{itemize}
Specifically, classical mechanisms like Laplace and Gaussian~\cite{DBLP:journals/fttcs/DworkR14}
add unbounded noise, resulting in an output domain of $(-\infty, +\infty)$.
This output domain makes them inapplicable to the circular domain $[0,2\pi)$.
Additionally, these mechanisms are designed for the distance metric $d(y,x) = |y-x|$, which is inconsistent with the distance in circular domain.
Discrete mechanisms, such as $k$-RR~\cite{DBLP:conf/nips/KairouzOV14}, assume a finite domain of size $k$ and therefore cannot be directly applied to the continuous circular domain $[0,2\pi)$.
Recently, OGPM~\cite{DBLP:journals/popets/ZhengMH25} proposed a piecewise-based mechanism tailored to circular domains, which inspires our design.

To highlight the necessity and superiority of \oursd's piecewise-based design
compared to simpler alternatives,
we start by presenting a strawman approach that extends the $k$-RR mechanism to the circular domain $[0,2\pi)$.

\textbf{Strawman approach: $k$-RR + uniform sampling.}
A former work ATP~\cite{DBLP:journals/pvldb/Zhang000H23} divides the direction space $[0,2\pi)$ into $k$ sectors
and applies the $k$-RR mechanism.
This approach ensures LDP for the $k$ sectors but not for the entire $[0,2\pi)$ domain.
% as any inference attack attempting to guess the sensitive direction $\varphi$ has at least a $1/k$ probability of success, regardless of $\varepsilon$. 
Even so, we can extend this approach to $[0,2\pi)$ by
further uniformly sampling a direction $\varphi'$ from
the output sector of $k$-RR.

Specifically, the circular domain $[0,2\pi)$ is divided into $k$ sectors
$[i-1, i)\cdot 2\pi/k$ for $i=1,2,\ldots,k$. 
Assume the sensitive direction $\varphi$ falls into the $j$-th sector.
Applying the $k$-RR mechanism outputs sector $i = j$
with probability $p = \exp(\varepsilon) / (k-1+\exp(\varepsilon))$,
or outputs sector $i \neq j$ with probability $p/\exp(\varepsilon)$.
Then, we uniformly sample a direction $\varphi'$ from the $i$-th sector,
i.e. $\varphi' \sim U(i-1, i)\cdot 2\pi/k$.

\textbf{Limitation.}
This strawman approach introduces inherent \emph{inner-sector} errors due to uniform sampling within a specific sector.
That is, even though the perturbed direction $\varphi'$ falls into the same sector as the sensitive direction $\varphi$,
the distance between $\varphi$ and $\varphi'$ may still be large due to the large sector size.
This issue is particularly prominent when $k$ is small, as each sector then spans $2\pi/k$ radians. 
Uniform sampling within such wide sectors introduces substantial inner-sector errors, which persist regardless of how large the privacy parameter $\varepsilon$ is set. 
On the other hand, increasing $k$ to reduce sector width impairs the utility of the $k$-RR mechanism~\cite{DBLP:conf/uss/WangBLJ17}.

\textbf{Design rationale.}
To address this limitation, 
we design a direction perturbation mechanism in which the perturbed direction 
is independent of hyperparameters like $k$,
thereby avoiding inherent errors beyond those introduced by the privacy parameter $\varepsilon$.

We adapt the design of piecewise-based mechanisms for circular domains from OGPM~\cite{DBLP:journals/popets/ZhengMH25}.
Specifically, (i) we instantiate a piecewise-based mechanism over the circular domain $[0,2\pi)$,
which guarantees LDP for the entire direction space;
(ii) the perturbed direction is sampled from a piecewise probability distribution,
which is centered around the sensitive direction with high probability and depends solely on the privacy parameter $\varepsilon$,
thus eliminating the need for pre-defined fixed sectors.

\begin{definition}[Direction Perturbation Mechanism] \label{def:direction_perturbation}
    Given a sensitive direction $\varphi$ and a privacy parameter $\varepsilon$,
    \oursd's direction perturbation mechanism $\mathcal{M}_\circ: [0,2\pi)\to [0,2\pi)$
    is defined by:
    \begin{equation*}
        pdf[\mechanism_\circ(\varphi)=\varphi'] =
        \begin{dcases}
            p_{\varepsilon} & \text{if} \ \varphi' \in [l_{\varphi,\varepsilon}, r_{\varphi,\varepsilon}), \\
            p_{\varepsilon} / \exp{(\varepsilon)} & \text{otherwise},
        \end{dcases}
    \end{equation*}
    where $p_{\varepsilon} = \frac{1}{2\pi} \exp(\varepsilon/2)$ is the sampling probability,
    and $[l_{\varphi,\varepsilon}, r_{\varphi,\varepsilon})$ is the sampling interval that
    \begin{equation*}
        \begin{split}
            l_{\varphi,\varepsilon} &= \left( \varphi - \pi\frac{\exp(\varepsilon/2) - 1}{\exp(\varepsilon) - 1} \right) \mod 2\pi, \\
            r_{\varphi,\varepsilon} &= \left( \varphi + \pi\frac{\exp(\varepsilon/2) - 1}{\exp(\varepsilon) - 1} \right) \mod 2\pi.
        \end{split}
    \end{equation*}
\end{definition}

The above mechanism $\mathcal{M}_\circ$ is defined by a piecewise probability distribution with three pieces,
where the central piece $[l_{\varphi,\varepsilon}, r_{\varphi,\varepsilon})$ has a higher probability density $p_{\varepsilon}$.
As a piecewise-based mechanism, it evidently satisfies LDP for the circular domain $[0,2\pi)$;
see Appendix~\ref{appendix:ldp_tracs-d} for the proof.

% Besides the whole-domain privacy guarantee, 
A key advantage of \oursd's direction perturbation mechanism is the ``dominant'' sector.
We refer to the central piece $[l_{\varphi,\varepsilon}, r_{\varphi,\varepsilon})$ as the dominant sector because it
centers around the sensitive direction $\varphi$ with high probability.
The dominant sector adapts dynamically to the privacy parameter $\varepsilon$:
as $\varepsilon$ increases, the dominant sector becomes narrower and more tightly centered around the sensitive direction $\varphi$.
This adaptivity effectively mitigates the inner-sector errors inherent in the strawman approach, where the sector width is fixed and independent of $\varepsilon$.

\begin{example} \label{example:direction_perturbation}
    Figure~\ref{fig:direction_perturbation} compares the strawman approach and \oursd.
    Assume the next location $\tau_{i+1}$ is a sensitive location needing direction perturbation,
    thus the sensitive direction is $\varphi: \tau_i' \to \tau_{i+1}$.
    The dominant sector of \oursd is much smaller as $\varepsilon$ increases.
    For instance, when $\varepsilon=6$ and $\varphi = \pi/6$, it can be calculated that $p_{\varepsilon} \approx 3.20$ and
    $[l_{\varphi,\varepsilon}, r_{\varphi,\varepsilon}) \approx [0.119\pi,0.214\pi)$.
    Thus, the dominant sector has a size of $0.095\pi \approx 17^\circ$ with probability $0.095\pi \times 3.20 \approx 0.95$ of being chosen.
    In comparison, the strawman approach has a dominant sector of $2\pi/6 = 60^\circ$ with probability $0.98$ of being chosen.
    With almost the same probability of being chosen, 
    \oursd's dominant sector is significantly narrower than that of the strawman approach, 
    resulting in a much smaller inner-sector (sampling) error.
    \end{example}

% \begin{example}
%     When $\varepsilon=4$ and $\varphi = \pi/6$, it can be calculated that $p_{\varepsilon} \approx 1.18$ and
%     $[l_{\varphi,\varepsilon}, r_{\varphi,\varepsilon}) \approx [0.05\pi,0.29\pi)$.
%     Thus, the dominant sector has a size of $0.24\pi \approx 42^\circ$ with probability $0.24\pi \times 1.18 \approx 0.88$
%     of being chosen.
%     \ In comparison, the strawman approach with $k=6$ has a dominant sector of $2\pi/6 = 60^\circ$ with probability $0.91$
%     of being chosen.
%     \ Figure~\ref{fig:direction_perturbation} illustrates the comparison between them.
%     Assume the next location $\tau_{i+1}$ is a private location needing direction perturbation,
%     thus the private direction is $\varphi: \tau_i' \to \tau_{i+1}$.
%     The dominant sector of \oursd is much smaller than the strawman approach,
%     especially when $\varepsilon$ becomes larger, e.g. $\varepsilon=6$
%     in Figure~\ref{fig:direction_perturbation_sub2}.
% \end{example}

Detailed comparison between the strawman approach and Definition~\ref{def:direction_perturbation}
is provided in Appendix~\ref{appendix:comparison_strawman}.
This Appendix will show that \oursd achieves a better trade-off between the size and probability of the dominant sector 
compared to the strawman approach with any $k$ value.

\subsubsection{Distance Perturbation} \label{subsubsec:distance_perturbation}

The other subspace of \oursd is the distance space $\mathcal{D}_{r(\varphi)}$.
Besides the boundedness, the size of this space also varies with the reference location $\tau_i$ and the direction $\varphi$.
The technical challenge in designing an LDP mechanism for this space lies in calculating 
the size $|\mathcal{D}_{r(\varphi)}|$ for any $\varphi$ at any $\tau_i$.

Given location $\tau_i$, any other location $\tau_j\in \traspace$ can be represented by a pair of coordinates $(\varphi, r(\varphi))$ relative to $\tau_i$.
Since we focus on $\traspace$ as a rectangular area $[a_{\text{sta}},a_{\text{end}}]\times[b_{\text{sta}},b_{\text{end}}]$,
the distance from $\tau_i$ to the boundary of $\traspace$ has a closed-form expression depending on the direction $\varphi$.
Specifically, this distance falls into one of the following four cases:
\begin{equation} \label{eq:distance_sapce}
    |\mathcal{D}_{r(\varphi)}| \in
    \left\{ \frac{a_{\text{end}} - a_i}{\cos(\varphi)},
    \frac{b_{\text{end}} - b_i}{\sin(\varphi)},
    \frac{a_i - a_{\text{sta}}}{-\cos(\varphi)},
    \frac{b_i - b_{\text{sta}}}{-\sin(\varphi)} \right\},
\end{equation}
depending on the value of $\varphi$.
Appendix~\ref{appendix:detailed_expression} provides the detailed equation form of Equation~(\ref{eq:distance_sapce}).
Figure~\ref{fig:distance_space} illustrates the first two cases of $|\mathcal{D}_{r(\varphi)}|$:
when $\varphi$ falls into $[0,\varphi_1)$, i.e. Figure~\ref{fig:distance_space_1},
$|\mathcal{D}_{r(\varphi)}| = (a_{\text{end}} - a_i) / \cos(\varphi)$;
Figure~\ref{fig:distance_space_2} shows the second case, i.e.
when $\varphi \in [\varphi_1,\varphi_2)$,
then $|\mathcal{D}_{r(\varphi)}| = (b_{\text{end}} - b_i) / \sin(\varphi)$.

\begin{figure}[t]
    \centering
    \begin{subfigure}[b]{0.43\linewidth}
        \centering
        \includegraphics[width=0.98\textwidth]{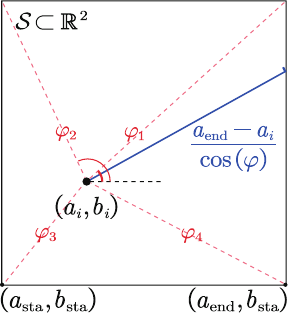}
        \caption{$\varphi \in [0,\varphi_1)$}
        \label{fig:distance_space_1}
    \end{subfigure}
    \hfill
    \begin{subfigure}[b]{0.43\linewidth}
        \centering
        \includegraphics[width=0.98\textwidth]{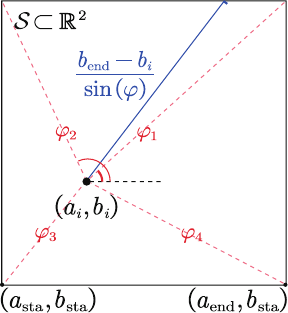}
        \caption{$\varphi \in [\varphi_1,\varphi_2)$}
        \label{fig:distance_space_2}
    \end{subfigure}
    \caption{Two cases of $|\mathcal{D}_{r(\varphi)}|$ at $\tau_i$ (blue lines).}
    \label{fig:distance_space}
\end{figure}

With the known distance space $[0,|\mathcal{D}_{r(\varphi)}|]$ and sensitive distance $r(\varphi)$,
we can design a distance perturbation mechanism to guarantee LDP for $\mathcal{D}_{r(\varphi)}$.
For the convenience of presentation, we normalize the distance $r(\varphi)$ by $|\mathcal{D}_{r(\varphi)}|$,
resulting in a normalized distance $\overline{r}(\varphi)\in [0,1)$.\footnote{
    We use $[0,1)$ instead of $[0,1]$ to ease the presentation of Definition~\ref{def:distance_perturbation}
    and to align with the circular domain. These two domains are equivalent in implementation.
}
% This normalization is a linear mapping and therefore preserves privacy by the post-processing property~\cite{DBLP:journals/fttcs/DworkR14}. 
We then employ the piecewise-based mechanism in OGPM~\cite{DBLP:journals/popets/ZhengMH25} for the normalized distance $\overline{r}(\varphi)$ to guarantee LDP on $[0,1)$.

\begin{definition}[Distance Perturbation Mechanism] \label{def:distance_perturbation}
    Given a sensitive distance $\overline{r}(\varphi)$ and a privacy parameter $\varepsilon$,
    \oursd's distance perturbation mechanism $\mathcal{M}_{-}: [0,1)\to [0,1)$
    is defined by:
    \begin{equation*}
        pdf[\mechanism_{-}(\overline{r}(\varphi))=\overline{r}'(\varphi)] =
        \begin{dcases}
            p_{\varepsilon} & \text{if} \ \overline{r}'(\varphi) \in [u, v), \\
            p_{\varepsilon} / \exp{(\varepsilon)} & \text{otherwise},
        \end{dcases}
    \end{equation*}
    where $p_{\varepsilon} = \exp(\varepsilon/2)$ is the sampling probability,
    and $[u,v)$ is the sampling interval that
    \begin{equation*}
        \begin{split}
            [u,v) &=
            \begin{cases}
                \overline{r}(\varphi) + [-C, C) & \text{if} \ \overline{r}(\varphi) \in [C, 1-C), \cr
                [0, 2C) & \text{if} \ \overline{r}(\varphi) \in [0, C), \cr
                [1-2C, 1) & \text{otherwise},
            \end{cases}
        \end{split}
    \end{equation*}
    with $C=(\exp(\varepsilon / 2) - 1) / (2\exp(\varepsilon) - 2)$.
\end{definition}

Similar to $\mathcal{M}_\circ$, the above mechanism $\mathcal{M}_{-}$ has a higher probability density $p_\varepsilon$ in the central piece $[u,v)$,
and a larger $\varepsilon$ results in a higher $p_\varepsilon$ and a narrower $[u,v)$, boosting the utility.
It ensures LDP for $[0,1)$ and outputs a perturbed normalized distance $\overline{r}'(\varphi)$.
To cooperate with the direction perturbation,
$\overline{r}'(\varphi)$ can be mapped to another distance space $\mathcal{D}_{r(\varphi')}$ at the perturbed direction $\varphi'$.
This is a linear mapping without randomness, so the post-processing property~\cite{DBLP:journals/fttcs/DworkR14}
ensures it preserves the same privacy level.

\subsubsection{Workflow of \oursd}

Combining the direction perturbation and distance perturbation, \oursd perturbs each sensitive location in the trajectory.

Algorithm~\ref{alg:tracs-d} presents the workflow of \oursd.
It takes the location space $\mathcal{S}$, the sensitive trajectory $\mathcal{T}$, and the privacy parameter $\varepsilon$ as input. 
Since \oursd relies on a non-sensitive reference location, we add a dummy location $\tau_0'$ to $\mathcal{T}$,
which can be the starting coordinate of $\mathcal{S}$ or a randomly drawn location.
Then, \oursd iterates over each pair of consecutive locations $\tau_i'$ and $\tau_{i+1}$ (line~3).
For each sensitive location $\tau_{i+1}$, it perturbs the direction (red block)
and normalized distance (blue block) with privacy parameters $\varepsilon_d$ and $\varepsilon - \varepsilon_d$, respectively.
% We set $\varepsilon_d = \varepsilon\pi / (\pi + 1)$ to balance the direction and distance perturbation.
The perturbed distance is then mapped along the perturbed direction (line~9).
Line~11 updates the reference location to the perturbed location $\tau_{i+1}'$, which is then used as the reference for the next iteration ($i+1$). 
Finally, the algorithm outputs the perturbed trajectory $\mathcal{T}'$. 

Algorithm~\ref{alg:tracs-d} uses each $\tau_i'$ as the reference location for the next perturbation.
In fact, it can be any non-sensitive location in $\traspace$.
We choose different $\tau_i'$ for alleviating the impact of a specific reference location.

\begin{algorithm}[t]
    \small
    \caption{\oursd}\label{alg:tracs-d}
    \SetKwComment{commentAlgo}{$\color{gray}\triangleright$\ }{}
    \KwIn{Rectangular location space $\mathcal{S}$, sensitive trajectory $\mathcal{T}=\{\tau_1, \tau_2, \dots, \tau_n\}$,
    privacy parameter $\varepsilon$}
    \KwOut{Perturbed trajectory $\mathcal{T}'=\{\tau_1', \tau_2', \dots, \tau_n'\}$}
    $\mathcal{T}' \gets \emptyset, \mathcal{T} \leftarrow  \tau_0' \cup\mathcal{T}$\commentAlgo*{\textcolor{gray}{Add a dummy location $\tau_0'$}}
    \For{$i\leftarrow 0$ \KwTo $n-1$}
    {
        \commentAlgo{\textcolor{gray}{$\tau_i'$ is the ref location for this iteration}}
        $\tau_i' = (a_i,b_i), \tau_{i+1} = (a_{i+1},b_{i+1})$\;
        \commentAlgo{\textcolor{gray}{Sensitive direction}}
        \tikzmk{A}
        $\varphi \leftarrow \text{atan2}(b_{i+1}-b_i, a_{i+1}-a_i)$\;
        \commentAlgo{\textcolor{gray}{Perturb direction}}
        $\varphi' \leftarrow \mechanism_\circ(\varphi;\varepsilon_d)$\;
        \tikzmk{B}\boxit{mypink}
        \tikzmk{A}
        $R \leftarrow |\mathcal{D}_{r(\varphi)}|$ in Equation~(\ref{eq:distance_sapce})\;
        $\overline{r}(\varphi) \leftarrow  ||\tau_{i+1} - \tau_i'||_2 / R$\commentAlgo*{\textcolor{gray}{Sensitive (norm.) distance}}
        $\overline{r}'(\varphi) \leftarrow \mechanism_{-} (\overline{r}(\varphi);\varepsilon - \varepsilon_d)$\commentAlgo*{\textcolor{gray}{Perturb distance}}
        \tikzmk{B}\boxit{myblue}
        $R'\leftarrow |\mathcal{D}_{r(\varphi')}|$, $r'(\varphi') \leftarrow \overline{r}'(\varphi)\times R'$\commentAlgo*{\textcolor{gray}{De-normalize}}
        \commentAlgo{\textcolor{gray}{Transform back to (longitude, latitude)}}
        $\tau_{i+1}' = (a_i + r'(\varphi')\cos(\varphi'), b_i + r'(\varphi')\sin(\varphi'))$\;
        $\mathcal{T}' \leftarrow \mathcal{T}' \cup \tau_{i+1}'$\commentAlgo*{$\tau_{i+1}'$ is the next ref location}
    }
    \Return  $\mathcal{T}'$\;
\end{algorithm}

\subsubsection{Analysis of \oursd} \label{subsubsec:analysis_tracs-d}

This subsection analyzes the privacy, computational complexity, and utility of \oursd.

\begin{theorem}\label{thm:tracs-d}
    \oursd (Algorithm~\ref{alg:tracs-d}) satisfies $n\varepsilon$-LDP for the rectangular location space $\mathcal{S}$.
\end{theorem}

\begin{proof}(Sketch)
    The direction perturbation mechanism $\mechanism_\circ$ satisfies $\varepsilon_d$-LDP by its definition.
    For the distance perturbation, the randomness comes entirely from mechanism $\mechanism_{-}$,
    which satisfies $(\varepsilon - \varepsilon_d)$-LDP.
    The post-processing property preserves the same privacy level after linearly mapping the perturbed distance. 
    Then each location satisfies $\varepsilon$-LDP 
    by Composition Theorem~\ref{thm:sequential} (or by computing the 2D $pdf$ ratio in the LDP definition).
    Hence, the entire perturbed trajectory satisfies $n\varepsilon$-LDP.
    Appendix~\ref{appendix:proof_tracs-d} provides details.
\end{proof}

\textbf{Complexity.}
\oursd (Algorithm~\ref{alg:tracs-d}) has $\Theta(1)$ time complexity for each location,
as each line is $\Theta(1)$.
Thus, the entire algorithm has $\Theta(n)$ time complexity, where $n$ is the length of the trajectory.
The space complexity is also $\Theta(n)$ because it needs to store the perturbed trajectory $\mathcal{T}'$.

% The $\Theta(n)$ time complexity is optimal for per-location perturbation mechanisms.
% Among existing LDP mechanisms for discrete location space,
% NGram~\cite{DBLP:journals/pvldb/CunninghamCFS21} and ATP~\cite{DBLP:journals/pvldb/Zhang000H23} are $\Omega(n)$
% due to their necessary optimization steps.
% LDPTrace~\cite{DBLP:journals/pvldb/DuHZFCZG23} is $\Theta(cn)$ with $c$ being the number of grids.
% Meanwhile, it also has a large constant factor due to its model training procedure on a set of trajectories.

The $\Theta(n)$ time complexity is optimal for per-location perturbation mechanisms.
Among existing LDP mechanisms for discrete location spaces,
even if we ignore the complexity of the score function and sampling in the Exponential mechanism~\cite{DBLP:conf/innovations/NissimST12},
NGram~\cite{DBLP:journals/pvldb/CunninghamCFS21}
has $\Theta(tn)$ time complexity, where $t$ is the number of predefined ``gram''.
% even with pre-calculated score functions for the Exponential mechanism~\cite{DBLP:conf/innovations/NissimST12},
% NGram~\cite{DBLP:journals/pvldb/CunninghamCFS21}
% has $\Theta(tn)$ time complexity, where $t$ is the size of the predefined ``gram'' space,
It also includes a large constant factor due to its search to satisfy reachability constraints.
L-SRR~\cite{DBLP:conf/ccs/WangH0QH22} has $\Theta(mn)$ time complexity, where $m$ is the number of locations, due to the grouping of locations.
ATP~\cite{DBLP:journals/pvldb/Zhang000H23} also has $\Theta(mn)$ time complexity,
because of its trajectory merging step.

Note that the $\Theta(n)$ space complexity is for the convenience of presenting the algorithm.
It can be reduced to $\Theta(1)$ by only updating each perturbed location $\tau_{i+1}'$
in place within $\mathcal{T}$.
This results in a lightweight memory footprint, which is crucial for edge computing.

\begin{theorem} \label{thm:tracs-d-variance}
    In \oursd, both the worst-case mean square error (MSE) of $\mechanism_\circ$ and $\mechanism_{-}$ converge to zero with a rate of $\Theta(e^{-\varepsilon/2})$.
    % Additionally, $\mechanism_\circ(\varphi)$ is unbiased and has a fixed variance for all $\varphi\in [0,2\pi)$.
\end{theorem}

\begin{proof}(Sketch)
    The MSE can be calculated by the closed-form expression of $\mechanism_\circ$ and $\mechanism_{-}$,
    allowing us to prove the convergence rate.
    % The unbiasedness and fixed variance of $\mechanism_\circ$ are due to the fact that it is always centered around the private direction $\varphi$.
    Appendix~\ref{appendix:proof_variance} provides the details.
\end{proof}

This theorem indicates that both the error of the perturbed direction and distance are reduced exponentially with the privacy parameter $\varepsilon$.
\ Note that the size of the distance space also affects the variance.
Specifically, denote the output of $\mechanism_{-}$ as a random variable $Y$ and the sensitive input as $x$.
After the linear mapping, the perturbed distance is $Y \cdot |\mathcal{D}_{r}|$ and the sensitive distance is $x \cdot |\mathcal{D}_{r}|$, where $|\mathcal{D}_{r}|$ is the size of the distance space.
Thus, the MSE of the perturbed distance is $\text{MSE}[Y \cdot |\mathcal{D}_{r}|] = |\mathcal{D}_{r}|^2 \cdot \text{MSE}[Y]$. 
Therefore, a smaller $|\mathcal{D}_{r}|$ leads to a smaller MSE of the perturbed distance.

%% file: our_method_2.tex
\subsection{\oursc}

\oursc decomposes the rectangular location space $\mathcal{S}$
into two independent distance spaces $\mathcal{D}_a$ and $\mathcal{D}_b$,
as shown in Figure~\ref{fig:decomposition_c}.
\ Specifically, we can fix $(a_{\text{sta}}, b_{\text{sta}})$ as the reference location,
with longitude $a_{\text{sta}}$ as the $a$-axis and latitude $b_{\text{sta}}$ as the $b$-axis.
Then any location $\tau_i \in \mathcal{S}$ can be represented as a pair of Cartesian coordinates $(d_a, d_b)$,
where $d_a \in \mathcal{D}_a$ and $d_b \in \mathcal{D}_b$ are the distances from $\tau_i$ to the $a$-axis and $b$-axis, respectively.
This representation is unique and reversible: any pair of distances $(d_a, d_b)$ can be mapped back to 
a unique location $\tau_i \in \mathcal{S}$.
\ Thus, if $\tau_i$ is a sensitive location, 
an LDP mechanism applied to $\tau_i$ is essentially applied to $d_a$ and $d_b$.
By Composition Theorem~\ref{thm:sequential}, we can achieve LDP for $\tau_j$ by applying perturbation mechanisms for $d_a$ and $d_b$, respectively.

\subsubsection{Coordinates Perturbation}
In \oursc's representation of locations, $\mathcal{D}_a = [0, a_{\text{end}} - a_{\text{sta}})$ and $\mathcal{D}_b = [0, b_{\text{end}} - b_{\text{sta}})$ 
are independent of any specific location.
Therefore, the location perturbation of any $\tau_i$ can be performed independently in $\mathcal{D}_a$ and $\mathcal{D}_b$.

Following the same idea of distance perturbation in \oursd,
we normalize the distance $d_a$ and $d_b$ to $[0, 1)$ by dividing $|\mathcal{D}_a|$
and $|\mathcal{D}_b|$, respectively.
Then the mechanism $\mathcal{M}_{-}$ in Definition~\ref{def:distance_perturbation} 
provides $\varepsilon$-LDP for the normalized distances,
and the linear mappings back to $\mathcal{D}_a$ and $\mathcal{D}_b$ preserve the privacy guarantee.

Algorithm~\ref{alg:tracs-c} shows the workflow of \oursc.
It takes the rectangular location space $\mathcal{S}$, the sensitive trajectory $\mathcal{T}$, and the privacy parameter $\varepsilon$ as input,
and outputs the perturbed trajectory $\mathcal{T}'$.
For each location $\tau_i = (a_i, b_i)$ in $\mathcal{T}$, Line~4 computes its coordinates $(d_a, d_b)$ and 
normalizes them to $(\overline{d}_a, \overline{d}_b) \in [0,1)\times [0,1)$.
Lines~5-6 perturb the normalized coordinates $\overline{d}_a$ and $\overline{d}_b$ 
using $\mathcal{M}_{-}$ with privacy parameter $\varepsilon/2$.
Then Line~8 maps the perturbed normalized coordinates back to $\mathcal{D}_a$ and $\mathcal{D}_b$,
and converts them to the longitude-latitude representation $\tau'_i$.

\begin{algorithm}[t]
    \small 
    \caption{\oursc}\label{alg:tracs-c}
    \SetKwComment{commentAlgo}{$\color{gray}\triangleright$\ }{}
    \KwIn{Rectangular location space $\mathcal{S}$, sensitive trajectory $\mathcal{T}=\{\tau_1, \tau_2, \dots, \tau_n\}$, 
    privacy parameter $\varepsilon$}
    \KwOut{Perturbed trajectory $\mathcal{T}'=\{\tau_1', \tau_2', \dots, \tau_n'\}$}
    $\mathcal{T}' \leftarrow \emptyset$\;
    \For{$i\leftarrow1$ \KwTo $n$}
    {
        $\tau_i = (a_i,b_i)$\;
        \commentAlgo{\textcolor{gray}{Normalize coordinates}}
        $(\overline{d}_a, \overline{d}_b) \leftarrow (\frac{a_i - a_{\text{sta}}}{a_{\text{end}}-a_{\text{sta}}}, \frac{b_i - b_{\text{sta}}}{b_{\text{end}} - b_{\text{sta}}})$\;
        \tikzmk{C}
        $\overline{d}'_a \leftarrow \mechanism_{-}(\overline{d}_a;\varepsilon/2)$\commentAlgo*{\textcolor{gray}{Perturb coordinates}} 
        \tikzmk{D}\boxitcd{mypink}
        \tikzmk{C}
        $\overline{d}'_b \leftarrow \mechanism_{-}(\overline{d}_b;\varepsilon/2)$\;
        \tikzmk{D}\boxitcd{myblue}
        \commentAlgo{\textcolor{gray}{De-normalize to $(\mathcal{D}_a,\mathcal{D}_b)$}}
        $(d_a', d_b') \leftarrow (\overline{d}'_a |\mathcal{D}_a|, \overline{d}'_b |\mathcal{D}_b|)$, $\tau'_i = (a_{\text{sta}}, b_{\text{sta}}) + (d_a', d_b')$\;
        $\mathcal{T}' \leftarrow \mathcal{T}' \cup \{\tau'_i\}$\;
    }
    \Return  $\mathcal{T}'$\;
\end{algorithm}

\subsubsection{Analysis of \oursc}

Privacy and complexity of \oursc can be analyzed similarly to \oursd.

\begin{theorem}\label{thm:tracs-c}
    \oursc (Algorithm~\ref{alg:tracs-c}) satisfies $n\varepsilon$-LDP for the rectangular location space $\mathcal{S}$.
\end{theorem}

\begin{proof}(Sketch)
    Algorithm~\ref{alg:tracs-c} uses mechanism $\mathcal{M}_{-}$ twice, each with a privacy parameter $\varepsilon/2$.
    By Composition Theorem~\ref{thm:sequential} (or by computing the 2D $pdf$ ratio in the LDP definition), their composition satisfies $\varepsilon$-LDP.
    The subsequent linear mappings to $\mathcal{D}_a$ and $\mathcal{D}_b$ are post-processing steps
    that preserve the same privacy.
    Hence, each perturbed location satisfies $\varepsilon$-LDP, 
    and the whole trajectory satisfies $n\varepsilon$-LDP.
    Appendix~\ref{appendix:proof_tracs-c} provides details.
\end{proof}    

\textbf{Complexity.}
\oursc (Algorithm~\ref{alg:tracs-c}) has $\Theta(n)$ time complexity and $\Theta(n)$ space complexity,
where $n$ is the length of the trajectory.
\ The time complexity is $\Theta(n)$ because each line of Algorithm~\ref{alg:tracs-c} is $\Theta(1)$.
The space complexity is $\Theta(n)$ because it needs to store the perturbed trajectory $\mathcal{T}'$.
Note that the $\Theta(n)$ space complexity is for the convenience of presenting the algorithm.
It can also be reduced to $\Theta(1)$ by updating each perturbed location $\tau_{i}'$
in place within $\mathcal{T}$, which benefits edge devices with limited computation resources.

\subsection{Rounding to Discrete Space}

Although \ours is designed for continuous location spaces, it can be applied to any fine-grained discrete location space
by rounding the perturbed locations to the nearest discrete locations.
The rounding is a post-processing step that does not affect the privacy guarantee for the continuous space.

Specifically, if the discrete location space is a set of cells,
we can round the perturbed location to the cell that contains it.
Formally, assume $\mathcal{C} = \{c_1, \ldots, c_m\} \subseteq \mathcal{S}$ discretize $\mathcal{S}$ into $m$ cells.
If $\tau_i'$ is the perturbed location of $\tau_i$ in \ours, then the discretized \ours outputs
$c_j$ such that $\tau_i' \in c_j$.
Commonly, the cells are evenly divided, so the rounding costs $\Theta(1)$ time,
as $c_j$ can be determined by the coordinates of $\tau_i'$.

If the location space is a set of location points, we can round the perturbed location to the nearest location point.
Formally, assume $\mathcal{P} = \{p_1, \ldots, p_m\} \subseteq \mathcal{S}$ is a set of location points, 
and $\tau_i'$ is the perturbed location of $\tau_i$ in \ours.
Then, the discretized \ours outputs $\argmin_{p_j \in \mathcal{P}} ||\tau_i' - p_j||_2$.
The time complexity of rounding to the nearest location point is $\mathcal{O}(m)$ in the worst case.
In practice, this time complexity can be significantly reduced by using heuristic or approximation algorithms.

The candidate locations in a discrete space are typically public knowledge, which enables a heuristic selection of a better reference location for \oursd.
The key observation is that a smaller distance space $\mathcal{D}_{r(\varphi)}$ yields a smaller distance-perturbation error.
Therefore, it is preferable to choose a reference location that minimizes the size of $\mathcal{D}_{r(\varphi)}$ with respect to the candidate locations.

Specifically, the direction from the reference location to other locations is often concentrated in a dominant sector, as illustrated in Figure~\ref{fig:direction_perturbation}.
If such a sector covers most candidate locations (thus reducing rounding error) and is close to the boundary of the location space (thus shrinking the distance domain),
then the corresponding reference location is a good choice.

\subsection{Discussions and Extensions}

This subsection discusses the comparison between \oursd and \oursc,
Euclidean versus spherical geometry, technical considerations in \ours, 
and other discussions and extensions.

\subsubsection{Comparison of Privacy Between \oursd and \oursc}

Although \oursd and \oursc provide the same level of privacy quantified by $\varepsilon$-LDP---more specifically, location-level (or event-level) LDP as described 
in surveys~\cite{DBLP:journals/popets/MirandaPascualGPFS23,DBLP:journals/popets/BuchholzAWNK24}---they have different privacy interpretations.

In \oursd, the LDP guarantee is provided for the direction information and the subsequent distance information.
This means that when the use of \oursd and the perturbed trajectory $\mathcal{T}'$ is publicly known,
any observer can hardly infer the sensitive direction $\varphi: \tau'_{i}\to \tau_{i+1}$ 
and the sensitive distance $r(\varphi): |\tau'_{i}\to \tau_{i+1}|$ from the known $\tau'_{i}$ and $\tau'_{i+1}$.
\ Although the distance space $\mathcal{D}_{r(\varphi)}$ relies on a specific direction,
it does not leak the direction information, as the direction space $\pdomain = [0,2\pi)$ is independent of any location
and \oursd uses the perturbed direction.

In \oursc, the LDP guarantee is provided for the two independent distance spaces $\mathcal{D}_a$ and $\mathcal{D}_b$.
This means that when the use of \oursc and the perturbed trajectory $\mathcal{T}'$ is publicly known,
any observer can hardly infer the sensitive distances $d_a$ and $d_b$.

\subsubsection{Euclidean Geometry vs Spherical Geometry}

\ours is designed for a rectangular location space $\mathcal{S} \subset \mathbb{R}^2$ under Euclidean geometry, where distances are measured by the Euclidean metric.
In particular, \oursd requires computing the distance spaces $\mathcal{D}_{r(\varphi)}$, which are defined with respect to the Euclidean distance.
This choice is made for generality:
(i) for many continuous trajectories (e.g. from wearable sensors or indoor devices), Euclidean geometry is a natural choice; and
(ii) for city-scale GPS trajectories (e.g. in Chicago and Tokyo), the distortion induced by approximating geographic coordinates as Cartesian coordinates is typically negligible.

Under spherical geometry, distances are measured by the great-circle distance.
(i) For country-scale GPS trajectories, standard map projections such as UTM~\cite{latlon_utm_converter} can convert GPS coordinates into local Cartesian coordinates, after which Euclidean distance can be used.
(ii) For general spherical domains (e.g. the Earth's surface), we can treat longitude and latitude as two independent coordinates and redesign \oursc via independent composition of $(\mechanism_\circ, \mechanism_{-})$.
Specifically, for a location with GPS coordinates $\tau_i = (a_i, b_i) \in ([-\pi, \pi), [-\pi/2, \pi/2])$, i.e. longitude and latitude, we perturb longitude and latitude independently using $\mechanism_\circ$ and $\mechanism_{-}$, respectively.
This design does not involve Euclidean distance and aligns with the semantics of longitude and latitude as separate coordinates.

\subsubsection{Why Perturb Direction First in \oursd}

In \oursd, we perturb the direction first and then the distance.
This order is motivated by two considerations.
First, the direction space $\pdomain = [0, 2\pi)$ is location-independent, making it well suited for direct perturbation.
Second, the distance space $\mathcal{D}_{r(\varphi)}$ is always defined relative to a specific direction $\varphi$.
If we were to perturb the distance before the direction, then after perturbing the direction we would need to redefine the corresponding distance domain using the perturbed direction, which complicates the procedure.
Perturbing the direction first avoids this issue and yields a cleaner mechanism design.

\subsubsection{Impact of the Size of $\mathcal{S}$}

Intuitively, a larger location space $\mathcal{S}$ requires larger privacy parameters to achieve the same level of utility.
For \oursd and \oursc, this impact is reflected in the size of the direction space and distance spaces.

In \oursd, the direction space $\pdomain = [0,2\pi)$ is independent of the size of $\mathcal{S}$,
meaning that the MSE of the perturbed direction, $\mathrm{MSE}[\varphi']$, remains constant regardless of the size of $\mathcal{S}$.
\ The distance space $\mathcal{D}_{r(\varphi)}$ depends on the specific direction $\varphi$ and the size of $\mathcal{S}$.
In the worst case, $\mathcal{D}_{r(\varphi)}$ is the diagonal of $\mathcal{S}$.
We have shown that the MSE of $\mechanism_{-}$ is quadratic to the size of the distance space
after linear mapping in Section~\ref{subsubsec:analysis_tracs-d}.
Therefore, when $|\mathcal{D}_{r(\varphi)}|$ increases linearly, the expected error of the perturbed distance, 
i.e. $\mathrm{E}[|r'(\varphi) - r(\varphi)|]$, also increases linearly. 
\ Nonetheless, the error in the direction space $\pdomain$ affects the error in $\mathcal{D}_{r(\varphi)}$ 
in a non-linear way.
Specifically, under the Euclidean distance metric, 
this effect is $|\mathcal{D}_{r(\varphi)}|\cdot 2\sin(|\varphi - \varphi'|/2)$,
which is the chord length of the arc between $\varphi$ and $\varphi'$.

In \oursc, both distance spaces $\mathcal{D}_a$ and $\mathcal{D}_b$ are determined by the size of $\mathcal{S}$.
Since they use mechanism $\mechanism_{-}$ independently, 
their expected error increases linearly with the increase of $|\mathcal{D}_a|$ and $|\mathcal{D}_b|$.
Meanwhile, the error in $\mathcal{D}_a$ affects the error in $\mathcal{D}_b$ linearly.

\subsubsection{Comparison with Geo-indistinguishability} \label{subsubsec:geo-indistinguishability}

Geo-indistinguish\-ability~\cite{DBLP:conf/ccs/AndresBCP13} (Geo-Ind for brevity, also called metric privacy~\cite{DBLP:journals/popets/QiuLPSX25} in recent works) is a more general (or weaker) version of LDP. 
From the perspective of Geo-ind, LDP force the proximity-aware function $d(x, y)$ in Geo-ind to be $d(x, y) = 1$ for all $x, y \in S$, 
ignoring that different $x$ and $y$ have different $d(x, y)$ in the context of their locations. 
Conversely, if the upper bound $\max_{x, y \in S} d(x, y) \leq d^*$ holds in $S$, then TraCS also satisfies $d^*\varepsilon$-Geo-ind. 
However, this trivial satisfaction of $d^*\varepsilon$-Geo-ind is highly suboptimal, 
since it uses the global upper bound of $d(x, y)$ rather than the fine-grained, actual values of $d(x, y)$ for different location pairs of $x$ and $y$.

We omit trajectory utility evaluations against Geo-ind as it is a different privacy model from LDP.
Generally, a well-designed Geo-ind mechanism can achieve better data utility than LDP mechanisms.
For piecewise-shaped mechanisms, it is possible to tailor different pieces to account for varying $d(x, y)$ for different $x, y$, 
thereby satisfying Geo-ind while achieving better data utility than the current LDP mechanisms.

\subsubsection{Other LDP Mechanisms for Bounded and Continuous Space} \label{subsubsec:other_bounded_ldp}

In this paper, we employ utility-optimized piecewise-based mechanisms (OGPM)~\cite{DBLP:journals/popets/ZhengMH25} for bounded numerical domains to design the direction and distance perturbations in \ours.
Other LDP mechanisms for bounded numerical domains can also be incorporated into \ours, but most of them may require redesign to suit the direction perturbation,
as the circular domain $\pdomain = [0, 2\pi)$ has a different distance metric than linear domains.
One exception is the Purkayastha mechanism~\cite{DBLP:conf/ccs/WeggenmannK21},
which is designed for sphere $\mathbb{S}^{n-1}$, e.g. circular lines when $n=2$.
It can be an alternative to $\mathcal{M}_\circ$, but not be directly applicable to $S\subset\mathbb{R}^2$.
We don't leverage it for direction perturbation because of its sampling complexity (Section~3.4.2 in~\cite{DBLP:conf/ccs/WeggenmannK21}) and higher MSE~\cite{DBLP:journals/popets/ZhengMH25}.

Another category of LDP mechanisms for bounded numerical domains is truncated mechanisms, 
including truncated Laplace~\cite{DBLP:journals/tdp/000619,DBLP:journals/jpc/HolohanABA20} and truncated Gaussian mechanisms~\cite{DBLP:journals/jpc/ChenH24}.
These mechanisms modify Laplace and Gaussian distributions by redesigning distributions on a bounded domain, e.g. $[0,1)$.
However, truncated mechanisms cannot preserve the same privacy level as their original counterparts~\cite{DBLP:journals/tdp/000619,DBLP:journals/jpc/HolohanABA20,DBLP:journals/jpc/ChenH24}.
Determining their new privacy level is non-trivial and typically relies on numerical testing algorithms~\cite{DBLP:journals/jpc/HolohanABA20,DBLP:journals/jpc/ChenH24} 
rather than closed-form solutions.
Compared with truncated mechanisms, piecewise-based mechanisms are explicitly designed for LDP in bounded numerical domains,
making them more effective and easier to analyze.

% \subsubsection{More Extensions} \label{subsubsec:more_extensions}

% Appendix~\ref{appendix:more_extensions} provides two more extensions of \ours and experimental results,
% including (i) other piecewise-based mechanisms and (ii) other shapes of location space.
% The experimental evaluations are included in Appendix~\ref{appendix:resigned_sw_exp}.

\subsubsection{Other Piecewise-based Mechanisms} \label{appendix:other_piecewise}

In Definition~\ref{def:direction_perturbation} and Definition~\ref{def:distance_perturbation},
we use specific instantiations for the parameters $p$ and $[l, r)$ in the piecewise-based mechanism.
Besides these instantiations, \ours can be integrated with any other piecewise-based mechanisms~\cite{DBLP:conf/icde/WangXYZHSS019,DBLP:conf/sigmod/Li0LLS20,DBLP:journals/tmc/MaZW24,10735696,DBLP:journals/iotj/ZhaoZYWWLNL21}
by redesigning them.
For example, SW~\cite{DBLP:conf/sigmod/Li0LLS20,DBLP:journals/popets/ZhengMH25} is a famous mechanism proposed for distribution estimation;
it is defined on 
$[0,1)\to [-b, 1+b)$, with $b$ as a parameter.
It uses different $p$ and $[l, r)$, and employs a strategy of maximizing the 
mutual information between the perturbed and sensitive data to determine $b$.
\ SW can be redesigned~\cite{DBLP:journals/popets/ZhengMH25} for the direction perturbation and distance perturbation in \oursd and \oursc.
Appendix~\ref{appendix:resigned_sw} provides the detailed redesign of SW for them
and shows intuitive comparisons with Definition~\ref{def:direction_perturbation} and Definition~\ref{def:distance_perturbation},
and Appendix~\ref{appendix:resigned_sw_exp} provides a comparison.

\subsubsection{Other Shapes of Location Space}

The design of LDP mechanisms is space-specific,
implying that extensions to other shapes of location space are not straightforward.

\textbf{By post-processing.} 
A practical workaround for non-rectangular continuous domains is to post-process perturbed locations so they lie inside the target region. 
If a perturbed location falls outside the shape, 
project it to the nearest point inside the shape (for polygons this is the closest point on the boundary). 
Because the shape is public and projection is a data-independent post-processing step, 
it does not weaken the LDP guarantee. This approach is especially useful for irregular shapes or unions of multiple cities.

\textbf{Direct extension.}
Directly extending \ours to other shapes requires decomposing the domain into independent subspaces that are compatible with our perturbation primitives.
For \ours, the challenge lies in determining the size of each subspace,
i.e. $|\mathcal{D}_{\varphi}|$, $|\mathcal{D}_{r(\varphi)}|$, $|\mathcal{D}_a|$, and $|\mathcal{D}_b|$.
Among them, $|\mathcal{D}_{r(\varphi)}|$ is not apparent but can be calculated for rectangular location spaces.
\ Although calculating $|\mathcal{D}_{r(\varphi)}|$ is generally non-trivial for irregular shapes,
\ours can be extended to parallelogram location spaces by a linear transformation to a rectangular space.
Meanwhile, \oursd can be extended to circular \emph{regions}, which are widely used in relative-location representations (e.g. radar and sonar systems).
Extending \oursc to circular regions is more challenging, since such regions do not admit a natural Cartesian decomposition into independent coordinates.
Appendix~\ref{appendix:resigned_sw_exp} reports experimental results of \oursd on circular location spaces.

%% file: experiments.tex
\section{Evaluation} \label{sec:evaluation}

This section evaluates the performance of \ours on both continuous and discrete spaces. 

\subsection{Results on Continuous Space}
The strawman method---our extension of $k$-RR for direction perturb\-ation---can be combined with the distance perturbation mechanism in \oursd to ensure LDP in continuous spaces $\mathcal{S} \subset \mathbb{R}^2$.
For simplicity, we continue to refer to this extended approach as the strawman method, and compare it with \ours in continuous space.
We also include the planar Laplace mechanism~\cite{DBLP:conf/ccs/AndresBCP13} $+$ clipping (truncation) as a baseline.
This baseline perturbs each location by adding noise drawn from a 2D planar Laplace distribution, and then truncates the output to the location space $\mathcal{S}$ if it falls outside.\footnote{
    Appendix~\ref{appendix:2d_laplace} provides more details on the planar Laplace mechanism + clipping.
}

\subsubsection{Setup} \label{sec:continuous_exp_setup}

\ours requires the following parameters: the location space $\mathcal{S}$, the sensitive trajectory $\mathcal{T}$, 
the privacy parameter $\varepsilon$, and $\varepsilon_d$ for direction perturbation in \oursd. 
The strawman method also requires an instance of $k$ in $k$-RR.
We use configurations as follows for the experiments.
\begin{itemize}
    \item Location space $\mathcal{S}$: We consider synthetic and real-world settings. 
    (i) The default location space is $[0,1) \times [0,1)$.
    Without loss of generality, we consider a larger space $\mathcal{S} = [0,2)\times [0,10)$ for comparison.
    We generate $100$ random trajectories for each location space as synthetic datasets.
    Specifically, each trajectory is generated by randomly sampling a sequence of $100$ points in $\mathcal{S}$.
    This setup excludes dataset-specific effects to emphasize algorithmic performance.
    \ (ii) Areas of TKY and CHI: These two trajectory datasets were collected in Tokyo and Chicago,
    extracted from the Foursquare dataset~\cite{DBLP:journals/tsmc/YangZZY15} and Gowalla dataset~\cite{DBLP:conf/kdd/ChoML11}, respectively.
    Locations in TKY are within area $[139.47, 139.90)\times [35.51, 35.86)$ and CHI within area $[-87.9, -87.5)\times [41.6, 42.0)$.
    These two areas are visually shown in Figure~\ref{appendix:fig:location_spaces}.
    We use the first $100$ trajectories in each dataset.

    \item Privacy parameter $\varepsilon_d$: We set $\varepsilon_d = \varepsilon\pi / (\pi + 1)$ to heuristically balance the
    perturbation in $\mechanism_{\circ}$ and $\mechanism_{-}$ according to their domain sizes.\footnote{
        Appendix~\ref{appendix:epsilon_d} provides more discussion on setting $\varepsilon_d$.
    }
    \item $k$-RR: We set $k=6$ as the strawman method's default value, and consider $3$-RR and $12$-RR for comparison.
\end{itemize}
% We consider $\varepsilon$ from $2$ to $10$.
% For 2D data (in continuous space) or large number of locations (in discrete space),
% $\varepsilon=10$ is not too large.
% For example, the $k$-RR mechanism with $|\mathcal{X}| = 3600$ outputs the true location with probability $\approx 0.85$ in this case. 
% Meanwhile, a small $\varepsilon=2$ leads to a probability of $\approx 0.002$.
% Similar for the Exponential mechanism but depends on the score function.
We evaluate $\varepsilon$ in the range $[2,10]$. For 2D discrete spaces with large location sets,
this range represents a moderate privacy level.\footnote{
    In Definition~\ref{def:edldp}, the privacy parameter $\varepsilon$ is defined with respect to a specific domain $\mathcal{X}$; its practical strength therefore depends on the size of $\mathcal{X}$.
    }
For example, with $k$-RR mechanism and $|\mathcal{X}| = 3600$ locations,
$\varepsilon=10$ yields the true location output with probability $\approx 0.85$, while $\varepsilon=2$ results in a much lower probability of $\approx 0.002$. 
The Exponential mechanism exhibits similar behavior, though the exact probabilities depend on the score function.
In such cases where $|\mathcal{X}|$ is large, even though the probability ratio $\exp(\varepsilon)$ in the LDP definition is high, 
it remains difficult for an adversary to infer the true location with high probability due to the vast number of possible locations.

\textbf{Trajectory utility metrics.}
We follow the Euclidean distance to measure the utility of the perturbed trajectory~\cite{DBLP:journals/pvldb/Zhang000H23,DBLP:journals/pvldb/CunninghamCFS21,DBLP:conf/ccs/WangH0QH22}.
The distance between two locations $\tau_i$ and $\tau'_i$ is defined as $\|\tau_i - \tau'_i\|_2$.
Therefore, given a sensitive trajectory $\mathcal{T}$ and a perturbed trajectory $\mathcal{T}'$,
the average error (AE) among all locations in the trajectory is:
\begin{equation*}
    AE(\mathcal{T}, \mathcal{T}') = \frac{1}{|\mathcal{T}|}\sum_{i=1}^{|\mathcal{T}|} \|\tau_i - \tau'_i\|_2,
\end{equation*}
where $|\mathcal{T}|$ is the number of locations in the trajectory.
A smaller AE indicates better trajectory utility.
We compute the average AE across all trajectories in the dataset for comparison.
Although AE is not a perfect utility measure, e.g. it may not fully reflect the requirements of specific downstream tasks, 
it remains the commonly used metric for evaluating the utility of perturbed trajectories~\cite{DBLP:journals/popets/MirandaPascualGPFS23,DBLP:journals/eswa/ZhaoPC20,DBLP:journals/pvldb/Zhang000H23,DBLP:journals/pvldb/CunninghamCFS21,DBLP:conf/ccs/WangH0QH22}.
Moreover, many other trajectory utility metrics are directly or indirectly related to AE, such as the preservation of range queries and hotspots~\cite{DBLP:journals/pvldb/CunninghamCFS21,DBLP:journals/pvldb/Zhang000H23}.
For brevity, we defer these additional metrics to Section~\ref{appendix:other_metrics}.

\subsubsection{Error Comparison on Synthetic Datasets}

\begin{figure}[t]
    \centering
    \begin{subfigure}[b]{0.45\linewidth}
        \centering
        \includegraphics[width=0.98\textwidth]{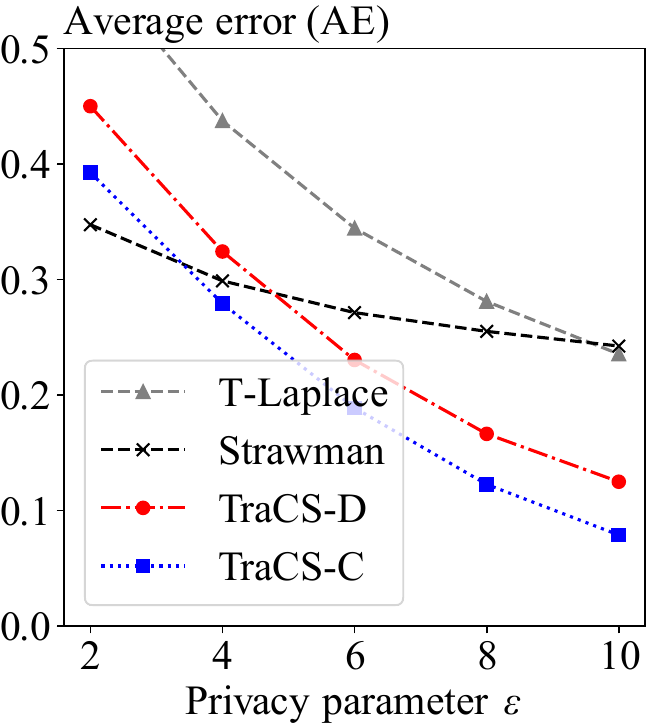}
        \caption{$\mathcal{S} = [0,1)\times[0,1)$}
        \label{fig:exp:1_1}
    \end{subfigure}
    \hfill
    \begin{subfigure}[b]{0.45\linewidth}
        \centering
        \includegraphics[width=0.98\textwidth]{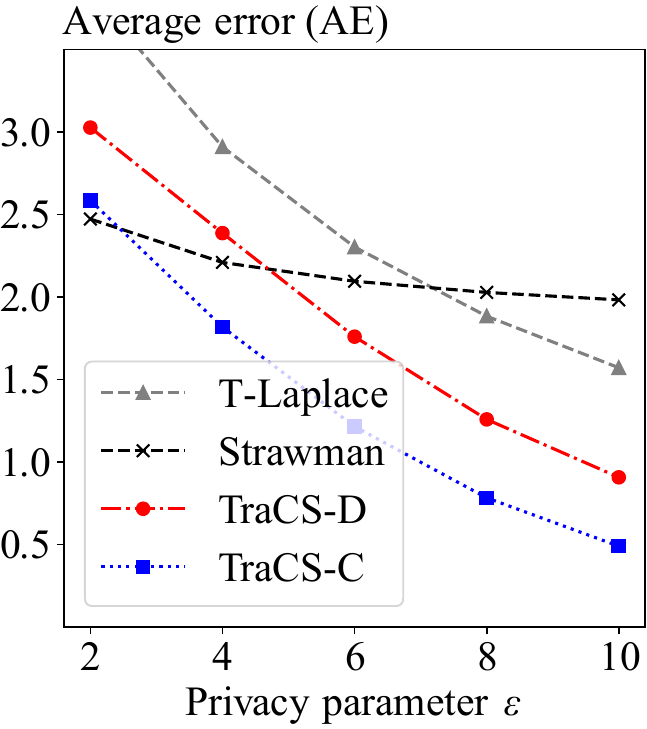}
        \caption{$\mathcal{S} = [0,2)\times[0, 10)$}
        \label{fig:exp:1_2}
    \end{subfigure}
    \caption{Comparison on synthetic datasets.}
    \label{fig:exp:1}
\end{figure}

\begin{figure}[t]
    \begin{minipage}{0.45\linewidth}
        \centering
        \includegraphics[width=0.98\textwidth]{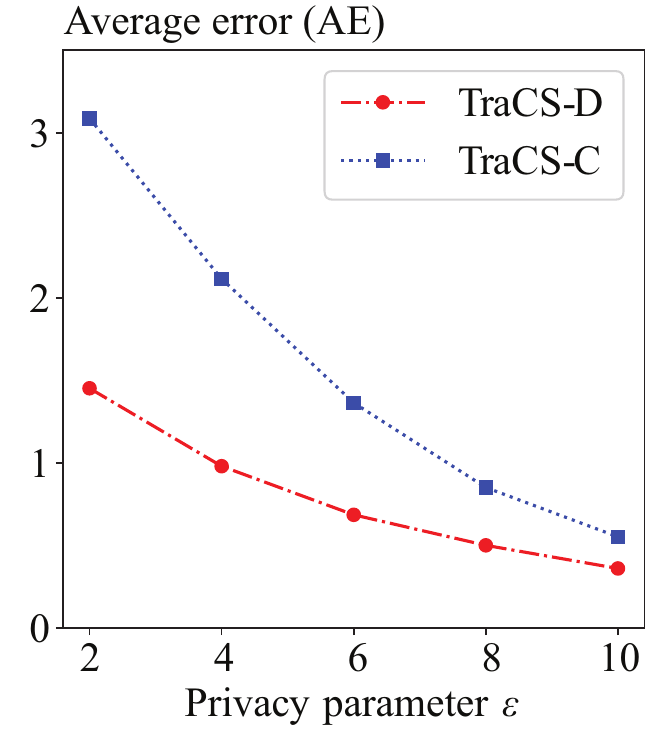}
        \caption{Condition \oursd outperforms \oursc.}
        \label{fig:exp:1_3}
    \end{minipage}
    \hfill
    \begin{minipage}{0.45\linewidth}
        \centering
        \includegraphics[width=0.98\textwidth]{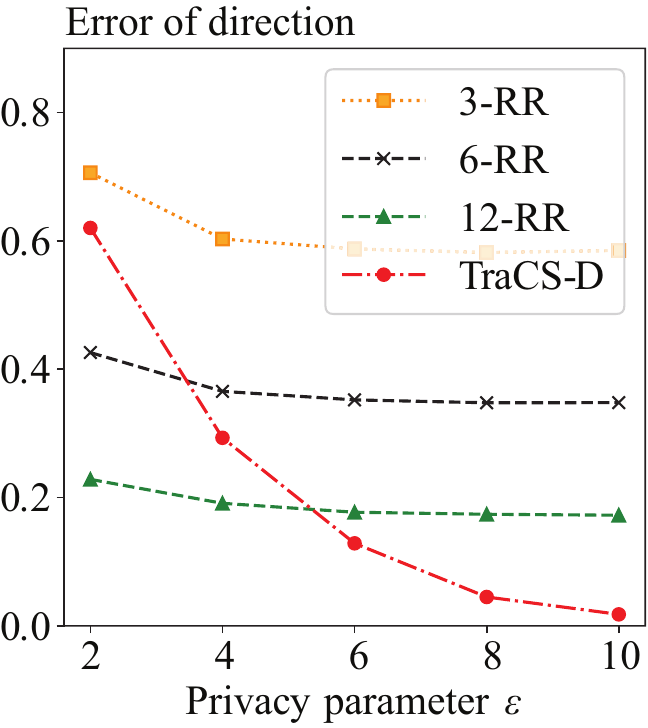}
        \caption{\oursd vs $k$-RR + uniform sampling.}
        \label{fig:exp:2}
    \end{minipage}
\end{figure}

Figure~\ref{fig:exp:1} presents the error comparison between \ours and the strawman method.
We can observe the error advantage of \ours over the truncated Laplace mechanism (T-Laplace) across all $\varepsilon$ values,
and over the strawman method as the privacy parameter $\varepsilon$ increases.
Specifically, in Figure~\ref{fig:exp:1_1}, \oursc exhibits smaller AE when $\varepsilon \gtrsim 3$,
and \oursd exhibits smaller AE when $\varepsilon \gtrsim 5$.
When the location space expands to $[0,2)\times [0,10)$, i.e. Figure~\ref{fig:exp:1_2},
the advantage of \ours becomes more significant.
\ Among \ours, \oursc consistently has a smaller AE than \oursd.
This is because \oursd often deals with larger distance spaces $\mathcal{D}_{r(\varphi)}$, which also magnifies the error of direction perturbation.
\ Statistically, the mean AE of \oursd is $91.1\%$ of the strawman method in Figure~\ref{fig:exp:1_1} and $86.6\%$ in Figure~\ref{fig:exp:1_2}
across all the $\varepsilon$ values.
For TraCS-C the corresponding ratios are $75.5\%$ in Figure~\ref{fig:exp:1_1} and $64.0\%$ in Figure~\ref{fig:exp:1_2}.

\textbf{Conditions \oursd outperforms \oursc.}
The error of \oursc depends on the size of the distance spaces $\mathcal{D}_a$ and $\mathcal{D}_b$,
while the error of \oursd depends only on the size of the distance space $\mathcal{D}_{r(\varphi)}$.
Therefore, if $\mathcal{D}_{r(\varphi)}$ is majorly smaller than $\mathcal{D}_a$ or $\mathcal{D}_b$,
\oursd will outperform \oursc.
\ We conduct experiments to verify this condition.
Specifically, we set $\mathcal{S} = [0,2)\times [0,10)$, and choose a sensitive trajectory along the $x$-axis, 
i.e. the short side of $\mathcal{S}$.
Then we perturb the trajectory $1000$ times using \oursc and \oursd and calculate the average AE.
In this case, $|\mathcal{D}_{r(\varphi)}|\approx 2$, much smaller than $|\mathcal{D}_b| = 10$.
Consequently, \oursd exhibits smaller AE than \oursc across all the $\varepsilon$ values,
as shown in Figure~\ref{fig:exp:1_3}.
In this experiment, the mean AE of \oursd is $49.8\%$ of \oursc across all the $\varepsilon$ values.

\textbf{Error of direction perturbation.}
Figure~\ref{fig:exp:2} shows the average error of direction perturbation in \oursd compared with the strawman method.
We collect the sensitive directions of the trajectories in the location space $\traspace = [0,1)\times [0,1)$, and perform direction perturbation using \oursd and the strawman method.
For the strawman method, we use $k$-RR with $k=3,6,12$.
This experiment reflects empirical error results of mechanism $\mechanism_{\circ}$ in Definition~\ref{def:direction_perturbation} and the $k$-RR + uniform sampling.
\ We can observe an exponentially decreasing error of \oursd as $\varepsilon$ increases,
as stated in Theorem~\ref{thm:tracs-d-variance}.
In contrast, the strawman method exhibits an eventually stable error as $\varepsilon$ increases, due to the inherent inner-sector error of $k$-RR.

\begin{figure}[t]
    \centering
    \begin{subfigure}[b]{0.465\linewidth}
        \centering
        \includegraphics[width=0.98\textwidth]{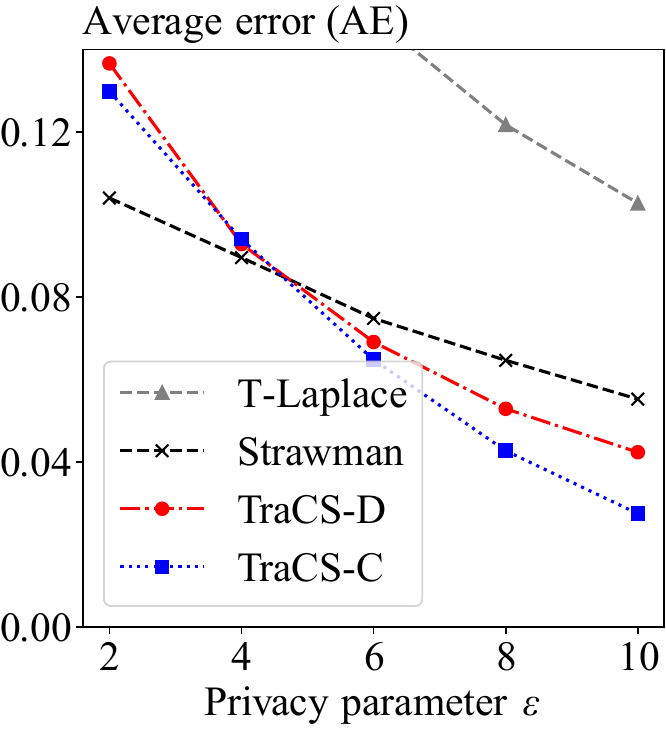}
        \caption{TKY dataset}
        \label{fig:exp:5_tky}
    \end{subfigure}
    \hfill
    \begin{subfigure}[b]{0.465\linewidth}
        \centering
        \includegraphics[width=0.98\textwidth]{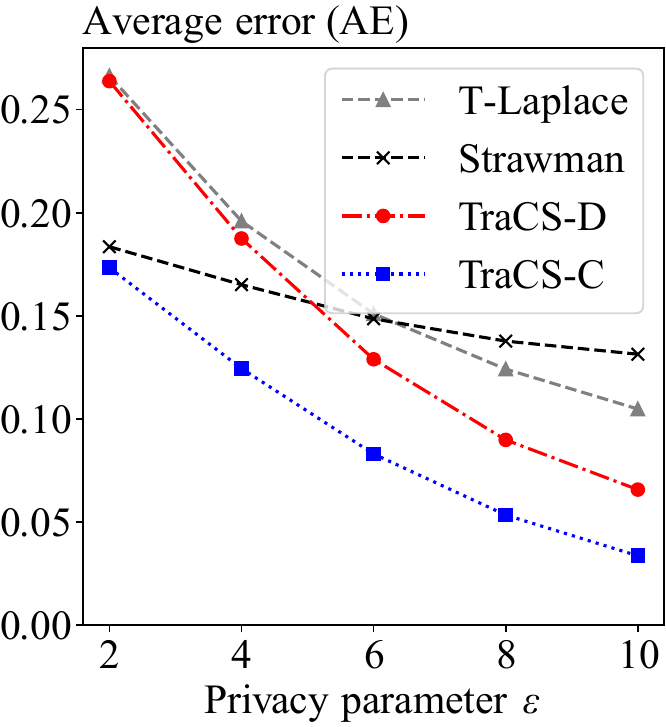}
        \caption{CHI dataset}
        \label{fig:exp:5_chi}
    \end{subfigure}
    \caption{Comparison on real-world datasets.}
    \label{fig:exp:5}
\end{figure}

\subsubsection{Error Comparison on Real-world Datasets}

Figure~\ref{fig:exp:5} shows the error comparison on the TKY and CHI areas. 
The trends largely mirror those observed on the synthetic datasets.
T-Laplace incurs substantially larger AE than \ours and the strawman method on the TKY dataset.
Across all $\varepsilon$ values, the mean AE of \oursd is $96.8\%$ of the strawman method on TKY and $94.5\%$ on CHI, 
for \oursc the corresponding ratios are $89.7\%$ and $61.2\%$. 
The performance gap is smaller on TKY than on CHI, because TKY's location space is more compact, 
which lowers the error values for all perturbation methods.

\subsubsection{Choose \oursd or \oursc}

From the above experiments, we can conclude criteria for choosing \oursd or \oursc:
(i) \oursc generally has better trajectory utility than \oursd for random trajectories;
(ii) \oursd is better for specific trajectories where their distance spaces are smaller than in \oursc,
and suitable for circular areas.

\subsection{Results on Discrete Space}

Our approach can be applied to discrete spaces by rounding each perturbed location to its nearest discrete point. 
Importantly, the privacy guarantee is established in the underlying continuous domain and therefore still holds after rounding, 
regardless of the chosen discretization. 
In the following, we evaluate the performance of \ours on discrete spaces and compare it with existing methods.

\subsubsection{Setup}

Our evaluation includes the following LDP methods for collecting trajectories in discrete spaces:

\begin{itemize}
    \item NGram~\cite{DBLP:journals/pvldb/CunninghamCFS21}: It is often impractical to 
    acquire reachability knowledge in practice~\cite{DBLP:journals/pvldb/Zhang000H23}.
    We instead consider a strong reachability constraint, i.e. from each location, the next location can only be within a distance of $0.75 \cdot \hat{a}$,\footnote{
        While the original NGram employs much stricter time-reachability constraints (e.g. the next location must be within $0.1$h $\cdot$ $4$km/h), this constraint is also \emph{known} to the adversary, which undermines the privacy guarantee.
        \ours are LDP mechanisms guaranteeing privacy for the entire domain and do not incorporate such constraints.
    } 
    where $\hat{a}$ is the maximal length of the location space.
    This constraint significantly reduces the location space for perturbation, especially for the locations near the boundary.              
    \item L-SRR~\cite{DBLP:conf/ccs/WangH0QH22}: We use group number $g=3$, the same as the empirical optimal value in their paper and code.
    Locations in the first group are within $0.3 \cdot \hat{a}$ distance from the sensitive location, 
    where $\hat{a}$ is the maximal length of the location space.
    The second group is within $0.6 \cdot \hat{a}$ distance from the sensitive location,
    and the remaining locations are in the third group.
    \item ATP~\cite{DBLP:journals/pvldb/Zhang000H23}: We set the $k$-RR parameter $k=6$ in ATP. 
    The privacy parameter $\varepsilon$ in original ATP is for the entire trajectory,
    we distribute it to each location's perturbation according to their algorithms.
    % The original ATP also distributes the privacy parameter $\varepsilon$ to
    % a complementary region constraining strategy to reduce the location space for perturbation.
    % For a clean comparison, we omit this strategy.
    % Although it uses a trajectory region constraining strategy to reduce the location space for perturbation, 
    % we omit this strategy and instead use reachability constraint in NGram to illustrate the effect.
    \item \ours: We use the same privacy parameter $\varepsilon_d$ as in the continuous space experiments.
    For each perturbed location from \ours, we round it to the nearest discrete location.
\end{itemize}

For a fair comparison, we use $\varepsilon$ as the privacy parameter for each location in all methods.
% NGram with reachability constraint also uses the same $\varepsilon$ for each location. We omit LDPTrace~\cite{DBLP:journals/pvldb/DuHZFCZG23} as it focuses on trajectory synthesis for a location space rather than perturbation for a given trajectory. 
The comparison is conducted on the following datasets:
\begin{itemize}
    \item Synthetic datasets: We uniformly discretize $[0,1)\times [0,1)$ into $m$ discrete points and treat the $m$ points as the discrete location space.
    Specifically, we set $m=10\times 10$ and $m=60\times 60$ for comparison.
    Along with them, we generate $100$ random trajectories with $100$ locations each as synthetic datasets,
    i.e. each location is randomly sampled from the $m$ discrete points and connected sequentially.
    \item TKY and CHI (visualized in Figure~\ref{appendix:fig:location_spaces}): We treat every location in the TKY and CHI datasets as a discrete point, which together constitute the discrete location space.
    TKY contains $7,798$ discrete locations; CHI contains $1,000$ discrete locations. We use the first $100$ trajectories for evaluation.
\end{itemize}

For the TKY and CHI datasets, we perform the perturbation $5$ times for each trajectory and compute the average AE,
which is the same setup as in ATP~\cite{DBLP:journals/pvldb/Zhang000H23}.

\subsubsection{Error Comparison on Synthetic Datasets}

\begin{figure}[t]
    \centering
    \begin{subfigure}[b]{0.45\linewidth}
        \centering
        \includegraphics[width=0.98\textwidth]{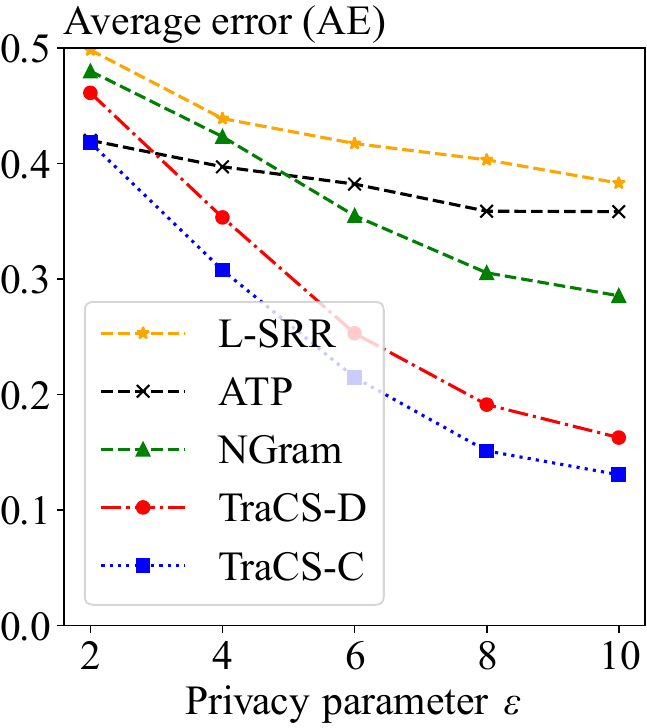}
        \caption{$m = 10\times 10$}
        \label{fig:exp:real:1_1}
    \end{subfigure}
    \hfill
    \begin{subfigure}[b]{0.45\linewidth}
        \centering
        \includegraphics[width=0.98\textwidth]{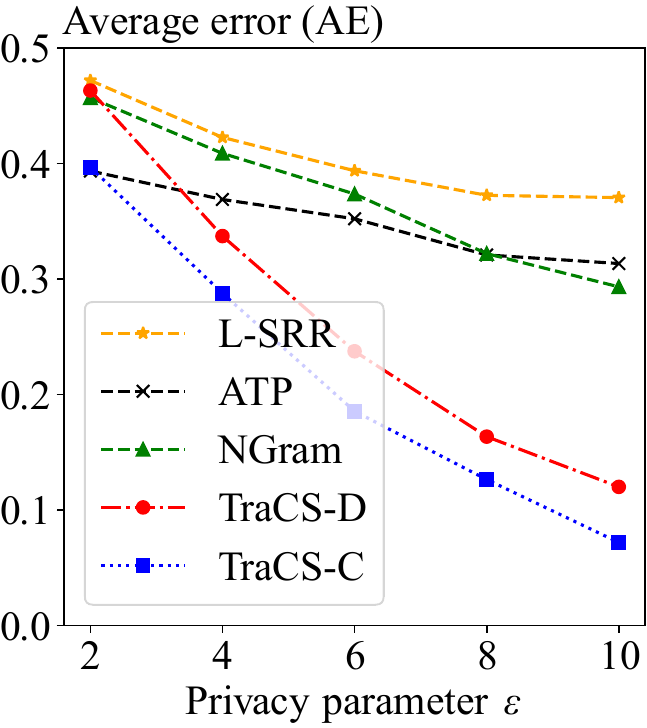}
        \caption{$m = 60\times 60$}
        \label{fig:exp:real:1_2}
    \end{subfigure}
    \caption{Comparison on synthetic datasets. $m$ is the number of discrete points in the location space.}
    \label{fig:exp:real_1}
\end{figure}

Figure~\ref{fig:exp:real_1} presents the error comparison on the synthetic dataset. Figure~\ref{fig:exp:real:1_1} shows results for $m=10\times 10$, 
and Figure~\ref{fig:exp:real:1_2} shows results for $m=60\times 60$. 
We observe that \oursc consistently outperforms NGram, L-SRR, and ATP across all the $\varepsilon$ values, 
and \oursd also outperforms them when $\varepsilon \gtrsim 3$. Moreover, the performance advantage of \ours becomes more significant as $m$ increases to $60\times 60$.
This is because finer-grained discretization allows \ours's rounding to align with the sensitive locations more accurately. Statistically, the mean AE of \oursc is $66.1\%$ of NGram,
$57.1\%$ of L-SRR, and $63.8\%$ of ATP in Figure~\ref{fig:exp:real:1_1}, 
and the proportions are $57.6\%$, $52.6\%$, and $61.1\%$ in Figure~\ref{fig:exp:real:1_2}. 
The mean AE of \oursd is $76.9\%$ of NGram, $66.4\%$ of L-SRR, and $74.2\%$ of ATP in Figure~\ref{fig:exp:real:1_1}, 
and the proportions are $71.2\%$, $65.0\%$, and $75.5\%$ in Figure~\ref{fig:exp:real:1_2}. 
These results indicate the advantage of \ours for synthetic trajectories in discrete spaces.

\subsubsection{Error Comparison on Real-world Datasets} \label{sec:real_discrete}
Figure~\ref{fig:exp:real_2} presents the error comparison on the TKY and CHI datasets. Compared to the synthetic dataset, NGram, L-SRR, and ATP achieve lower errors, particularly when $\varepsilon$ is small. 
However, as $\varepsilon$ increases, \ours still outperforms them.

\textbf{Reasons for the results.}
(i) ATP performs exceptionally well when $\varepsilon$ is small, particularly on the TKY dataset. This is because this dataset contains many ``condensed'' locations, i.e. no other locations near the trajectory. 
In such cases, if the bi-direction perturbation of ATP is accurate, the perturbed location is error-free.
(ii) The AE of ATP does not decrease significantly as $\varepsilon$ increases, which is consistent with their paper. 
NGram shows a similar trend, especially on the TKY dataset. 
This is due to the low performance of the Exponential mechanism for large location spaces, 
particularly when the locations are concentrated in a small area. 
In such cases, the score function struggles to differentiate the scores among the locations.
L-SRR exhibits a similar trend.
(iii) \ours exhibits larger AE than them when $\varepsilon$ is small because it focuses on the rectangular area formed by the longitude and latitude, instead of a set of discrete locations in the cities
(visualized in Figure~\ref{appendix:fig:location_spaces}).
Therefore, the perturbed locations may be far from the city center (where benefits discrete mechanisms) when $\varepsilon$ is small.

\begin{figure}[t]
    \centering
    \begin{subfigure}[b]{0.465\linewidth}
        \centering
        \includegraphics[width=0.98\textwidth]{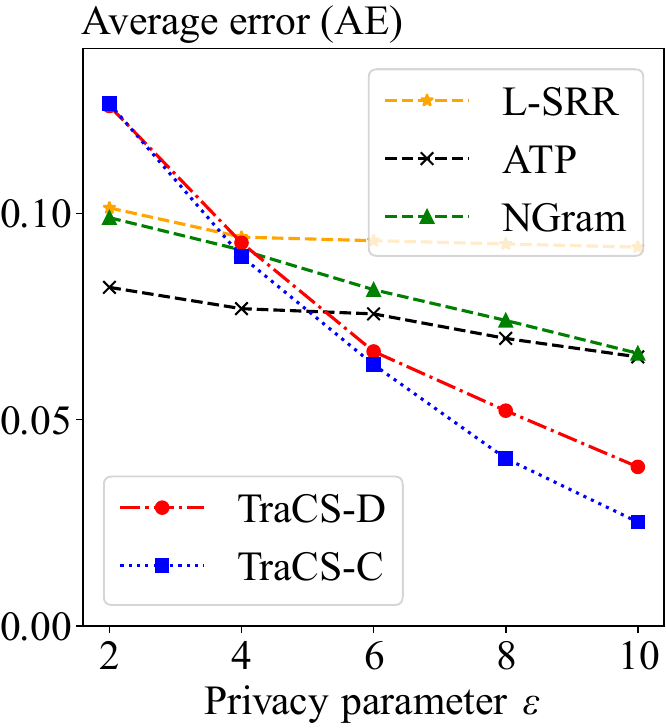}
        \caption{TKY dataset}
        \label{fig:exp:real:2_1}
    \end{subfigure}
    \hfill
    \begin{subfigure}[b]{0.465\linewidth}
        \centering
        \includegraphics[width=0.98\textwidth]{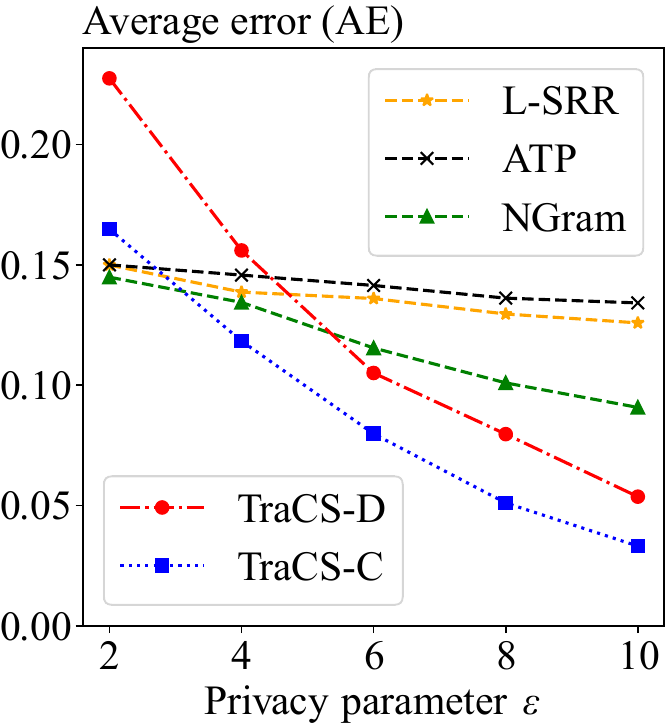}
        \caption{CHI dataset}
        \label{fig:exp:real:2_2}
    \end{subfigure}
    \caption{Comparison on real-world datasets. 
    \ours operates on rectangular areas encompassing the city, which may result in larger AE when $\varepsilon$ is small.}
    \label{fig:exp:real_2}
\end{figure}

\begin{table}[t]
    \centering
    \caption{Time cost comparison (in milliseconds).}
    \label{tab:time_cost}
    \setlength{\tabcolsep}{6pt} 
    \begin{tabular}{lrrr >{\columncolor{gray!10}}r >{\columncolor{gray!10}}r}
        \toprule
        & ATP & NGram & L-SRR & \oursd & \oursc \\
        \midrule
        \textbf{Total} & $145.7$ & $100.9$ & $6.2$ & $0.06$ & $0.05$ \\
        \midrule
        \textbf{Perturb} & $125.8$ & $92.8$ & $0.003$ & $0.018$ & $0.003$ \\
        \bottomrule
    \end{tabular}
\end{table}

\subsubsection{Time Cost Comparison}

Table~\ref{tab:time_cost} presents the time cost comparison on the TKY dataset. 
The time cost is measured in milliseconds and averaged per location. 
It shows that \oursd and \oursc are significantly faster than NGram, ATP, and L-SRR.
Compared to NGram and ATP, this efficiency is primarily due to the piecewise-based mechanisms in \ours, which require negligible time for perturbation compared to the Exponential mechanism.
Compared to L-SRR, \oursd and \oursc are also faster because they do not require grouping locations.

ATP has the highest time cost because it needs to process two copied trajectories and merges them.
NGram has a lower time cost than ATP but is still significantly higher than \ours. 
The perturbation procedure accounts for $86.3\%$ and $92.0\%$ of the total time cost in ATP and NGram, respectively.
L-SRR has the lowest time cost among the three state-of-the-art methods due to its efficient sampling mechanism,
however, its grouping procedure still incurs a non-negligible time cost.

The total time cost of \ours is less than $0.05\%$ of ATP and NGram, and less than $1\%$ of L-SRR.
The perturbation procedure of \oursd and \oursc only takes $30\%$ and $6\%$ time, respectively. 
\oursd has a higher time cost than \oursc because it needs to calculate the distance space $\mathcal{D}_{r(\varphi)}$ for each location in the trajectory, 
while \oursc can use the same $\mathcal{D}_a$ and $\mathcal{D}_b$ for all locations.

\subsubsection{Experimental Results for Other Metrics} \label{appendix:other_metrics}

This subsection evaluates the performance of \ours and existing methods under two additional metrics: range query preservation (Formula 17 in~\cite{DBLP:journals/pvldb/CunninghamCFS21}) and hotspot preservation (Section 6.2.4 in~\cite{DBLP:journals/pvldb/Zhang000H23}).
Given a threshold $\delta$, a perturbed location is considered correct if it lies within a $\delta$ distance of the corresponding sensitive location.
Formally, for a trajectory $\mathcal{T}$ and its perturbed trajectory $\mathcal{T}'$, the range query preservation (RQP) is defined as:
\[
RQP(\traj, \traj') = \frac{1}{|\mathcal{T}|}\sum_{i=1}^{|\mathcal{T}|} \mathbf{1}\{\|\tau_i - \tau'_i\|_2 \leq \delta\} \cdot 100\%,
\]
where $\mathbf{1}\{\cdot\}$ is the indicator function.
A higher RQP indicates better trajectory utility.
Hotspot preservation measures how many hotspots from a given set remain after perturbation.
Given a set of hotspots $\mathcal{H}$, the locations in $\mathcal{H}$ are considered preserved if their perturbed locations after rounding are also in $\mathcal{H}$.
Formally, the count difference (CD) of a trajectory in hotspot preservation is defined as:
\[
CD(\traj, \traj') = \sum_{i=1}^{|\mathcal{T}|} \mathbf{1}\{\tau_i \in \mathcal{H}\} - \sum_{i=1}^{|\mathcal{T}|} \mathbf{1}\{\tau_i \in \mathcal{H} \land \tau'_i \in \mathcal{H}\},
\]
where $\mathbf{1}\{\cdot\}$ is the indicator function.
A smaller CD indicates better trajectory utility.
\ We can find the relationship between AE and these two metrics.
For range query preservation with a given threshold $\delta$, when $AE \leq \delta$, the perturbed location is expected to satisfy the range query.
For hotspot preservation a smaller AE leads to better hotspot retention after rounding in \ours.

\begin{figure}[t]
    \centering
    \begin{subfigure}[b]{0.465\linewidth}
        \centering
        \includegraphics[width=0.98\textwidth]{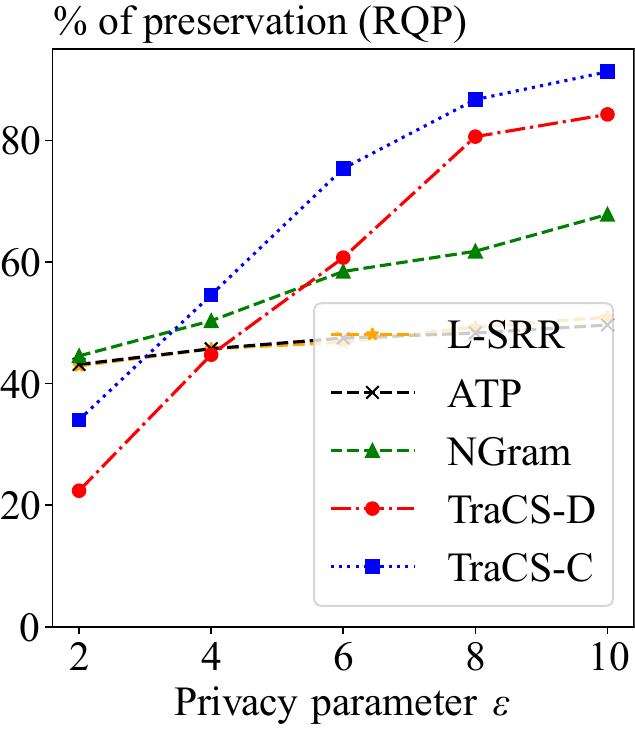}
        \caption{Range query preservation}
        \label{appendix:fig:range_query}
    \end{subfigure}
    \hfill
    \begin{subfigure}[b]{0.45\linewidth}
        \centering
        \includegraphics[width=0.98\textwidth]{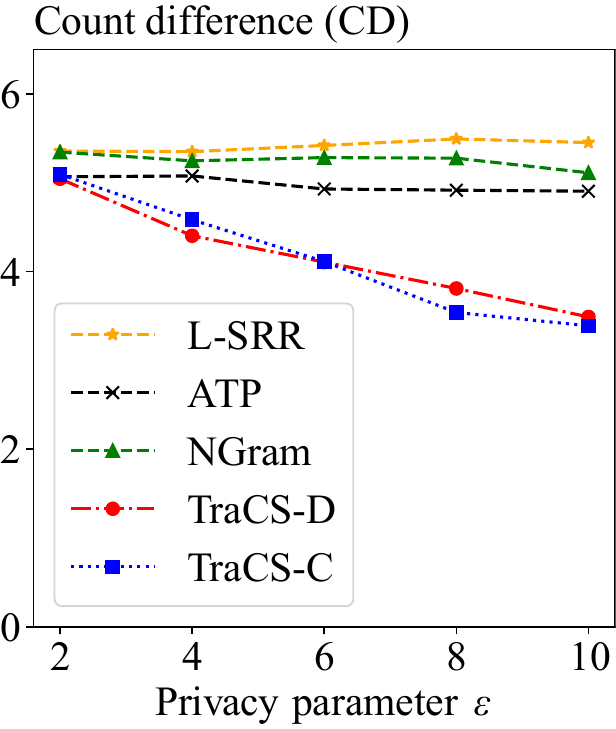}
        \caption{Hotspot preservation}
        \label{appendix:fig:hotspot}
    \end{subfigure}
    \caption{Comparison of other trajectory metrics on the CHI dataset. 
    Higher RQP and lower ACD indicate better performance.}
    \label{appendix:fig:other_metrics}
\end{figure}

We further evaluate the performance of \oursd and existing methods under varying privacy parameters $\varepsilon$ using these two metrics.
Specifically, we use the CHI dataset and adopt the following experimental setups:
\begin{itemize}
    \item Range query preservation: We set the threshold $\delta = 0.1$, and compute the RQP averaged over all trajectories.
    \item Hotspot preservation: The hotspot set $\mathcal{H}$ is defined as the first $20\%$ frequent locations in the CHI dataset and we compute the CD averaged over all trajectories.
\end{itemize}

The results are shown in Figure~\ref{appendix:fig:other_metrics}. 
It can be seen that \ours generally outperforms existing methods across different privacy levels,
demonstrating better effectiveness in preserving both range queries and hotspots.
Statistically, (i) the mean RQP across all $\varepsilon$ values for L-SRR, ATP, and NGram are 47.1\%, 46.8\%, and 56.5\%, respectively, 
while the mean RQP of \oursd and \oursc are 58.6\% and 68.4\%, respectively.
(ii) The mean CD across all $\varepsilon$ values for L-SRR, ATP, and NGram are 5.4, 4.9, and 5.2, respectively, 
while the mean CD of \oursd and \oursc are both 4.1.
Particularly, \ours shows improvement in hotspot preservation compared to existing methods across all privacy levels.

\subsubsection{Assigning $\varepsilon$ for a Whole Trajectory} \label{subsubsec:trajectory_level_epsilon}
We also evaluate \ours when the privacy parameter $\varepsilon$ is specified for an entire trajectory rather than per location.
Specifically, for \ours, L-SRR, NGram, and ATP, we set $\varepsilon$ at the trajectory level and allocate it uniformly across locations by assigning each location a parameter of $\varepsilon/|\mathcal{T}|$, where $|\mathcal{T}|$ is the trajectory length.
This uniform allocation is simple and requires no additional information, and we consider it the most robust choice.
In contrast, non-uniform allocations (e.g. assigning a larger parameter to utility-critical locations) may improve utility, 
but they typically rely on public notions of ``importance'' and can weaken privacy protection for those locations.
Appendix~\ref{appendix:trajectory_level_epsilon} provides results and discussion on this evaluation, 
showing that \ours outperforms existing methods in trajectory utility under this setting when $\varepsilon$ is large.

\subsection{Summary of Evaluation}

From the evaluation results on both continuous and discrete spaces, we can summarize the following findings:
\begin{itemize}
    \item In continuous spaces, \ours generally exhibits better trajectory utility than the strawman method,
    particularly when the privacy parameter $\varepsilon$ is large.
    Meanwhile, \oursc has better utility than \oursd for common-shape location space and trajectories.
    \item In discrete spaces, \ours generally outperforms NGram, L-SRR, and ATP,
    particularly when the privacy parameter $\varepsilon$ is large.
    Furthermore, \ours has significantly lower time cost, requiring less than $1\%$ of their time cost.
\end{itemize}
These findings validate the effectiveness and efficiency of \ours in collecting trajectory data under LDP in both continuous and discrete spaces.

\section{Conclusions}

This paper presents \oursd and \oursc, providing local differential privacy guarantees for trajectory collection in continuous location spaces.
\oursd uses a novel direction-distance perturbation procedure, while \oursc perturbs the Cartesian coordinates of each location.
These two methods can also be applied to discrete location spaces by rounding each perturbed location to the nearest discrete point, 
achieving better efficiency than existing discrete methods.
Trajectory utility of \ours is analyzed theoretically and evaluated empirically.
Evaluation results on discrete location spaces show that \ours outperforms state-of-the-art discrete methods in trajectory utility, particularly when the privacy parameter is large.

In summary, \ours demonstrates that collecting trajectory data
in continuous spaces generally better discrete approaches,
avoiding privacy and efficiency issues caused by
discretization, while can also be rounded to discrete spaces when necessary.

%% file: related_work.tex
\section{Related Work} \label{sec:related_work}
\setcounter{footnote}{0}

% This paper focuses on trajectory collection under LDP, a mathematically quantifiable privacy notion.
This section reviews related trajectory collection methods under LDP and other quantifiable privacy notions.
More privacy-preserving trajectory collection methods can be found in recent surveys~\cite{DBLP:journals/popets/MirandaPascualGPFS23,DBLP:journals/popets/BuchholzAWNK24}.

\subsubsection*{Trajectory Collection under LDP}

% Based on the privacy notion, trajectory collection methods under LDP can be categorized into two groups: pure LDP and relaxed LDP.
% Compared with relaxed LDP, pure LDP imposes a stronger constraint for indistinguishability.

We have disscussed three state-of-the-art trajectory collection methods under pure LDP:
NGram~\cite{DBLP:journals/pvldb/CunninghamCFS21}, L-SRR~\cite{DBLP:conf/ccs/WangH0QH22}, and ATP~\cite{DBLP:journals/pvldb/Zhang000H23}.
Among them, L-SRR was originally applied for origin-destination location pairs,
while NGram, ATP, and this paper focus on trajectory perturbation.
\ NGram attempts to ensure pure LDP for the ``trajectory space'', 
also named as trajectory-level or instance-level LDP in surveys~\cite{DBLP:journals/popets/MirandaPascualGPFS23,DBLP:journals/popets/BuchholzAWNK24},
which is a stronger privacy than the location space.
However, this is proved extremely challenging as the trajectory space increases exponentially with the number of possible locations in the location space.
L-SRR, ATP and this paper focus on ensuring pure LDP for the location space.
Unlike L-SRR, ATP and other existing methods based on discrete mechanisms, 
this paper focuses on continuous location spaces, which are more expressive than any fine-grained discrete location space.

% As a special case of trajectory collection, L-SRR~\cite{DBLP:conf/ccs/WangH0QH22} focuses on single-location publishing.
% It discretizes the location space into different-sized cells, with smaller cells near the sensitive location and larger cells farther away.
% Then, SRR (a variant of $k$-RR) is applied to the cells.
% Similar to other discretization methods, the number of cells affects its privacy and utility.
% \ours provides solutions to this problem in continuous spaces, achieving better privacy and utility, particularly when the number of cells is large.

Beyond the above three methods that design LDP mechanisms for trajectory collection, 
several works leverage external knowledge to improve the utility of LDP-based trajectory collection.
WF-LDPSR~\cite{DBLP:journals/compsec/LiXZZ25} adaptively determines each user's privacy protection level based on the sensitivity of the user's information, 
and applies a water-filling strategy to allocate privacy budgets during perturbation.
Regional popularity (i.e. the popularity of different regions in the location space) can also be exploited to constrain trajectory boundaries and improve utility~\cite{DBLP:journals/tifs/ZhangNFHZ25}.
In addition, recent work studies poisoning attacks against NGram and ATP~\cite{DBLP:journals/corr/abs-2503-07483}, 
where adversaries induce malicious users to submit poisoned trajectories to the data collector with the goal of promoting target patterns.

Another line of research is trajectory \emph{synthesis} for privacy-preser\-ving trajectory publication.
A typical work is LDPTrace~\cite{DBLP:journals/pvldb/DuHZFCZG23},
which synthesizes trajectories using a generative model.
LDPTrace discretizes the location space into cells and represents each trajectory as a sequence of cells.
Transition patterns between neighboring cells are perturbed using LDP mechanisms and collected by an untrusted curator.
The curator then trains a Markov chain model on the perturbed transition patterns to capture movement behaviors.
Following LDPTrace, ADGTrace~\cite{DBLP:journals/tmc/CaiLSZYZX26} also employs Markov chain models for trajectory synthesis; 
it does not provide LDP guarantees and leverages personal features to improve utility.
Compared to methods that perturb each trajectory, such as NGram, L-SRR, and ATP,
synthetic trajectories may be overly too random and not reflect the original trajectories if not personalized enough.

\subsubsection*{LDP Mechanisms for Bounded Numerical Domains}
The perturbation mechanisms used in \ours for direction and distance perturbation are utility-optimized piecewise-based mechanisms in OGPM~\cite{DBLP:journals/popets/ZhengMH25}.
Other LDP mechanisms for bounded numerical domains~\cite{DBLP:conf/ccs/WeggenmannK21,DBLP:journals/tdp/000619,DBLP:journals/jpc/HolohanABA20} can also be incorporated into \ours; refer to Section~\ref{subsubsec:other_bounded_ldp} for discussions and comparisons.
Compared with OGPM, which focuses on optimizing the 1D mechanisms' data utility, the design of 2D mechanisms in \ours needs to consider the unique characteristics of trajectory data and 2D spaces.
In particular, the direction-distance perturbation in \oursd to guarantee LDP for a rectangular region requires careful design.
Meanwhile, this paper emphasizes the overlooked fundamental issues inherent in discrete-space LDP mechanisms for trajectory collection~\cite{DBLP:journals/pvldb/CunninghamCFS21,DBLP:journals/pvldb/Zhang000H23,DBLP:conf/ccs/WangH0QH22},
which make privacy, utility, and efficiency discretization-dependent and prevent guarantees for continuous spaces.
\ours demonstrates that designing mechanisms directly in continuous space (via building blocks like OGPM, etc.) avoids these issues and subsumes discrete cases via rounding.

\subsubsection*{Trajectory Collection under DP}
As a weaker privacy notion than LDP, (central) DP~\cite{DBLP:conf/icalp/Dwork06,DBLP:journals/fttcs/DworkR14}
assumes a trusted curator who processes the raw data and releases the perturbed data.
It requires that only neighboring locations in the location space be indistinguishable, rather than any two arbitrary locations.
Compared to LDP, the concept of ``neighboring'' in DP has many interpretations.
A comprehensive survey on DP-based trajectory collection and publication is available in~\cite{DBLP:journals/popets/MirandaPascualGPFS23}.

The use of direction information in trajectory collection dates back to an early work SDD~\cite{DBLP:conf/ssdbm/JiangSBKT13},
which attempts to add unbounded Laplace noise to the direction information.
It provides a DP guarantee for unbounded continuous location spaces, but results in poor utility.
The notion of ``n-gram'' in trajectory collection
dates back to~\cite{DBLP:conf/ccs/ChenAC12}, where it abstracts a trajectory pattern as a gram.
The idea of hierarchical decomposition can also be found in~\cite{DBLP:conf/sigmod/ZhangXX16},
where a spatial decomposition method is used for trajectory synthesis.

External knowledge is also widely used in DP-based trajectory collection methods.
NGram's prior work~\cite{DBLP:conf/ssd/CunninghamCF21} provides a DP-based solution for trajectory synthesis
by exploiting the public knowledge of the road network.
Focusing on particular privacy-concerned locations, or sensitive zones~\cite{DBLP:journals/compsec/WangK20}, can also
improve the trajectory utility while maintaining partial privacy guarantees.

\subsubsection*{Trajectory Collection under Other Privacy Notions}
A widely used relaxed privacy notion than LDP for trajectory collection is Geo-indistinguishability~\cite{DBLP:conf/ccs/AndresBCP13,DBLP:journals/popets/MendesCV20},
where the indistinguishability between two locations is a function of their distance, i.e. privacy can decrease with distance.
Input-discriminative LDP~\cite{DBLP:conf/icde/Gu0XC20} goes one step further, assigning each location its own privacy level.
Threshold-integrated LDP~\cite{10.1007/978-981-97-5562-2_5} defines a threshold to constrain the sensitive region near a location.
This is similar to NGram's reachability constraint but with a formal privacy definition.
\ More relaxed LDP notions for trajectory collection can be found in paper survey~\cite{DBLP:journals/popets/MirandaPascualGPFS23}.
Although these relaxed LDP notions generally provide better trajectory utility than pure LDP, 
they sacrifice the strong privacy guarantee of pure LDP, and often need hyperparameters to define the privacy level.

%% file: appendix.tex
\section{Proofs} \label{appendix:proofs}

\subsection{LDP Proof of Definition~\ref{def:direction_perturbation}} \label{appendix:ldp_tracs-d}

The proof applies to all piecewise-based mechanisms, including the direction perturbation mechanism in Definition~\ref{def:direction_perturbation} 
and the distance perturbation mechanism in Definition~\ref{def:distance_perturbation}.

\begin{proof}
    For any two sensitive directions $\varphi_1$, $\varphi_2$ and any $\varphi' \in [0, 2\pi)$, we have:
    \begin{equation*}
        \frac{pdf[\mechanism_\circ(\varphi_1) = \varphi']}{pdf[\mechanism_\circ(\varphi_2) = \varphi']} \leq \frac{p_\varepsilon}{p_\varepsilon / \exp(\varepsilon)} = \exp(\varepsilon).
    \end{equation*}
    The inequality holds because any $\varphi'$ is sampled with a probability of at most $p_\varepsilon$ and at least $p_\varepsilon / \exp(\varepsilon)$, 
    regardless of whether the input is $\varphi_1$, $\varphi_2$ or any other value.
    % Therefore, the direction perturbation mechanism $\mechanism_\circ$ satisfies $\varepsilon$-LDP.
\end{proof}

\subsection{Proof of Theorem~\ref{thm:tracs-d}} \label{appendix:proof_tracs-d}

\begin{proof}
    Denote \oursd in Algorithm~\ref{alg:tracs-d} as $\mechanism$.
    We need to prove that $\mechanism$ satisfies $\varepsilon$-LDP for each location $\tau \in \traspace$,
    i.e. 
    $$
    \forall \tau_1, \tau_2, \tau' \in \mathcal{S}: \frac{pdf[\mathcal{M}(\tau_1) = \tau']}{pdf[\mathcal{M}(\tau_2) = \tau']} \leq \exp(\varepsilon).
    $$
    In \oursd, each location $\tau$ is represented in a direction-distance space with unique coordinates $(\varphi, r(\varphi))$.
    The key insight of the proof is that the perturbation of $\varphi$ and $r(\varphi)$ satisfies $\varepsilon_d$-LDP and $(\varepsilon - \varepsilon_d)$-LDP, respectively.
    Then \oursd ensures $\varepsilon$-LDP for the location space $\traspace = \mathcal{D}_\varphi \otimes \mathcal{D}_{r(\varphi)}$.
    We provide proofs by the Composition Theorem~\ref{thm:sequential} and by computing the 2D $pdf$ ratios in the LDP definition.

    \textbf{By Composition Theorem~\ref{thm:sequential}.} The definitions of piecewise-based mechanisms $\mechanism_\circ$ and $\mechanism_{-}$ imply that they satisfy LDP.
    Specifically, for $\mechanism_\circ(\varphi;\varepsilon_d)$, 
    $$
    \forall \varphi_1, \varphi_2, \varphi' \in [0, 2\pi): \frac{pdf[\mechanism_\circ(\varphi_1) = \varphi']}{pdf[\mechanism_\circ(\varphi_2) = \varphi']} \leq \exp(\varepsilon_d),
    $$
    is guaranteed by Definition~\ref{def:direction_perturbation} of $\mechanism_\circ$.
    Similarly, $\mechanism_{-}(\overline{r}(\varphi);\varepsilon - \varepsilon_d)$ in Definition~\ref{def:distance_perturbation} ensures $(\varepsilon - \varepsilon_d)$-LDP
    for $\overline{r}(\varphi)$.
    For clarity, let $\overline{r} \coloneqq \overline{r}(\varphi)$, which yields
    $$
    \forall \overline{r}_1, \overline{r}_2, \overline{r}' \in [0, 1): \frac{pdf[\mechanism_{-}(\overline{r}_1) = \overline{r}']}{pdf[\mechanism_{-}(\overline{r}_2) = \overline{r}']} \leq \exp(\varepsilon - \varepsilon_d).
    $$
    The subsequent linear mapping from $\overline{r}\in [0,1)$ to $r \in \mathcal{D}_{r(\varphi')}$ is a post-processing step that does not affect the LDP property.
    Specifically, denote the linear mapping as $g:[0,1)\to \mathcal{D}_{r(\varphi')}$, and denote $r_1 = g(\mechanism_{-}(\overline{r}_1))$ and $r_2 = g(\mechanism_{-}(\overline{r}_2))$ as two random variables.
    It is clear that
    $$
    \forall r_1, r_2, r' \in \mathcal{D}_{r(\varphi')}: \frac{pdf[r_1 = r']}{pdf[r_2 = r']} \leq \exp(\varepsilon - \varepsilon_d),
    $$
    because the linear mapping $g$ is a deterministic function without randomness.
    Thus, $\mechanism_{-}$ ensures $(\varepsilon - \varepsilon_d)$-LDP for $\mathcal{D}_{r(\varphi')}$.

    Combining these results with Sequential Composition Theorem~\ref{thm:sequential}, 
    we can conclude that $\mechanism = (\mechanism_\circ, \mechanism_{-})$ ensures $\varepsilon$-LDP for each location $\tau = (\varphi, r(\varphi)) \in \traspace$.
    % i.e. the probability ratio of getting any location $\tau^*$ from $\tau_1$ and $\tau_2$ is bounded by $\exp(\varepsilon)$.
    Consequently, the entire trajectory with $n$ locations satisfies $n\varepsilon$-LDP.

    \textbf{By computing the 2D \bm{$pdf$} ratio.}
    Assuming two sensitive locations $\tau_1=(\varphi_1,\overline{r}_1)$ and $\tau_2=(\varphi_2,\overline{r}_2)$,
    where $\overline{r}_1$ and $\overline{r}_2$ are normalized distances from the reference location along $\varphi_1$ and $\varphi_2$, respectively.
    For any output $\tau' = (\varphi',\overline{r}')$, the 2D $pdf$ of getting it from $\tau_1$ and $\tau_2$ are
    \begin{equation*}
        \begin{split}
            pdf[\mathcal{M}(\tau_1) = \tau'] = pdf[\mathcal{M}_\circ (\varphi_1)=\varphi']\cdot pdf[\mathcal{M}_-(\overline{r}_1)=\overline{r}'], \\
            pdf[\mathcal{M}(\tau_2) = \tau'] = pdf[\mathcal{M}_\circ (\varphi_2)=\varphi']\cdot pdf[\mathcal{M}_-(\overline{r}_2)=\overline{r}'],
        \end{split}
    \end{equation*}
    where $\mathcal{M}_-$ is applied to the \emph{same} normalized distance space (from the reference location to the boundary of $S$ along $\varphi'$). 
    The ratio is bounded by $\exp(\varepsilon_d)\cdot \exp(\varepsilon - \varepsilon_d) = \exp(\varepsilon)$ due to the definition of $\mathcal{M}_\circ$ and $\mathcal{M}_-$. 
    Consequently, the entire trajectory with $n$ locations satisfies $n\varepsilon$-LDP.
\end{proof}

\subsection{Proof of Theorem~\ref{thm:tracs-d-variance}} \label{appendix:proof_variance}

\begin{proof}
    We prove the $\Theta(e^{-\varepsilon/2})$ convergence rate by calculating the MSE of the perturbation mechanism.
    Specifically, the MSE of $\mechanism_{-}$ is calculated as follows:
    \begin{equation*}
        \mathrm{MSE}[\mechanism_{-}(\overline{r})] = \int_{0}^{1} (\overline{r}^* - \overline{r})^2 pdf[\mechanism_{-}(\overline{r}) = \overline{r}^*] d\overline{r}^*.
    \end{equation*}
    The worst-case MSE is achieved when $\overline{r} = 0$ or $\overline{r} = 1$, i.e. at the endpoints (similar to the variance calculation in~\cite{DBLP:conf/icde/WangXYZHSS019}). 
    In this case,
    \begin{equation*}
        \begin{split}
            \mathrm{MSE}[\mechanism_{-}(0)] & = \int_{0}^{2C} \left(\overline{r}^* - 0\right)^2 p_{\varepsilon} d\overline{r}^* + \int_{2C}^{1} \left(\overline{r}^* - 0\right)^2 \frac{p_\varepsilon}{e^\varepsilon}d\overline{r}^* \\
                                            & = \frac{8C^3 p_\varepsilon}{3} + \frac{(1-8C^3)p_\varepsilon}{3e^\varepsilon}.
        \end{split}
    \end{equation*}
    with $p_{\varepsilon}$ and $C$ defined in Definition~\ref{def:distance_perturbation}.
    We omit the calculation for $\overline{r} = 1$ as it is symmetric to $\overline{r} = 0$.
    Plugging in the values of $p_{\varepsilon}$ and $C$, we can simplify the MSE using big $\mathcal{O}$ notation:
    \begin{equation*}
        \begin{split}
            \mathrm{MSE}[\mechanism_{-}(0)] &= \frac{1}{3}e^{-\varepsilon/2} + \frac{(e^{\varepsilon/2}-1)^3}{3e^{\varepsilon/2}(e^{\varepsilon} - 1)^2} \\
            &= \frac{1}{3}e^{-\varepsilon/2} + \mathcal{O}(e^{-\varepsilon}).
        \end{split}
    \end{equation*}
    Therefore, as $\varepsilon \to \infty$, the worst-case MSE of $\mechanism_{-}$ converges to zero at a rate of $\Theta(e^{-\varepsilon/2})$.
    The MSE of $\mechanism_{\circ}$ can be calculated in the same way, resulting in a convergence rate of $\Theta(e^{-\varepsilon/2})$ as well.
\end{proof}

\subsection{Proof of Theorem~\ref{thm:tracs-c}} \label{appendix:proof_tracs-c}

\begin{proof}
    Denote \oursc in Algorithm~\ref{alg:tracs-c} as $\mechanism$. We show that $\mechanism$ satisfies $\varepsilon$-LDP for each location $\tau \in \traspace$.
    In \oursc, each location $\tau$ is represented in the Cartesian space $\mathcal{D}_a \times \mathcal{D}_b$ with unique coordinates $(a, b)$.
    Thus, it suffices to prove that the perturbations of $a$ and $b$ each satisfy $\varepsilon/2$-LDP.

    From Theorem~\ref{thm:tracs-d}, $\mechanism_{-}(\overline{a}; \varepsilon/2)$ satisfies $\varepsilon/2$-LDP for the normalized coordinate $\overline{a} \in [0, 1)$.
    The subsequent linear mapping from $\overline{a} \in [0, 1)$ to $a \in \mathcal{D}_a$ is post-processing and therefore preserves the LDP guarantee.
    Hence, $\mechanism_{-}$ ensures $\varepsilon/2$-LDP over $\mathcal{D}_a$.
    By the same argument, $\mechanism_{-}$ also ensures $\varepsilon/2$-LDP over $\mathcal{D}_b$.

    Therefore, $\mechanism \coloneqq (\mechanism_{-}, \mechanism_{-})$ satisfies $\varepsilon$-LDP for each location $\tau = (a, b) \in \traspace$ by the Sequential Composition Theorem~\ref{thm:sequential}.
    Equivalently, one can verify the guarantee by directly bounding the 2D $pdf$ ratio induced by $\mechanism$.
    Consequently, the entire trajectory with $n$ locations satisfies $n\varepsilon$-LDP.
\end{proof}

\section{Complementary Materials} \label{appendix:complementary_materials}

\subsection{Space-Dependence of Indistinguishability (Section~\ref{sec:introduction})} \label{appendix:space_dependence}

Indistinguishability of a data point is always relative to a specified data space, which determines the set of alternatives that an adversary may try to distinguish it from.
Accordingly, indistinguishability cannot be meaningfully discussed without first specifying the underlying space, for two main reasons.

\textbf{(i) LDP perspective.} The LDP guarantee is defined with respect to a particular input domain.
If two algorithms operate on the same domain and use the same privacy parameter $\varepsilon$, then their privacy guarantees are comparable (in the sense of the LDP definition);
if their domains differ, then their guarantees may not be directly comparable, even when they share the same $\varepsilon$.

\textbf{(ii) Information-theoretic perspective.} For an $\varepsilon$-LDP mechanism $\mechanism$, the mutual information $I(x;\mechanism(x))$ is upper bounded by a function of both $\varepsilon$ and the size of the input space,
i.e. $I(x;\mechanism(x)) \leq \mathcal{O}(\varepsilon, |\mathcal{X}|)$, where $|\mathcal{X}|$ denotes the cardinality of the input space $\mathcal{X}$.
A similar bound can be found in~\cite{Duchi2013LocalPD}.\footnote{
    Intuitively, LDP enforces \emph{pairwise} indistinguishability but does not directly account for how indistinguishability accumulates over many alternatives, which can be captured by mutual information.
}

This observation is particularly relevant to trajectory collection or synthesis under LDP when discretization-based methods are used.
Even for a fixed continuous area, different discretization strategies (e.g. uniform grids, adaptive grids, or different sets of points of interest (POIs)) induce different discrete location spaces with different sizes and spatial layouts.
Consequently, the effective indistinguishability provided by the same LDP mechanism can vary substantially across discretizations, even when the same $\varepsilon$ is used.

\subsection{Limitations of the Exponential Mechanism (Section~\ref{sec:introduction} and Section~\ref{sec:existing})} \label{appendix:exponential_mechanism}

\subsubsection{High Computational Complexity} \label{appendix:exp_mechanism_complexity}

The biggest limitation of the Exponential mechanism is the complexity in calculating the score function and sampling from the probability distribution.
Before perturbing location $x$, the Exponential mechanism requires computing the score function $d(x,y)$ for every possible location $y$ in the location space.
If the location space has $m$ locations, the time complexity of computing the score function for 
a single location is $\mathcal{O}(m)$, which is computationally expensive for large-scale spaces.
\ Moreover, ATP's dynamic strategy to constrain the location space for perturbation makes
it need to recompute the score function for each location,
resulting in $\mathcal{O}(m)$ time complexity, where $m$ is the size of the location space.

The complexity of sampling from the Exponential mechanism is also $\mathcal{O}(m)$,
as its cumulative distribution function (CDF) is an $m$-piece function.
Sampling requires comparing a random number with the CDF values of all $m$ locations.

In contrast, we have seen that piecewise-based perturbation mechanisms in \ours do not require score functions,
and they sample from a $3$-piece CDF at each perturbation,
resulting in a negligible constant time complexity, i.e. $\Theta(1)$, for each perturbation.

\begin{figure}
    \centering
    \begin{subfigure}[b]{0.46\linewidth}
        \centering
        \includegraphics[width=\textwidth]{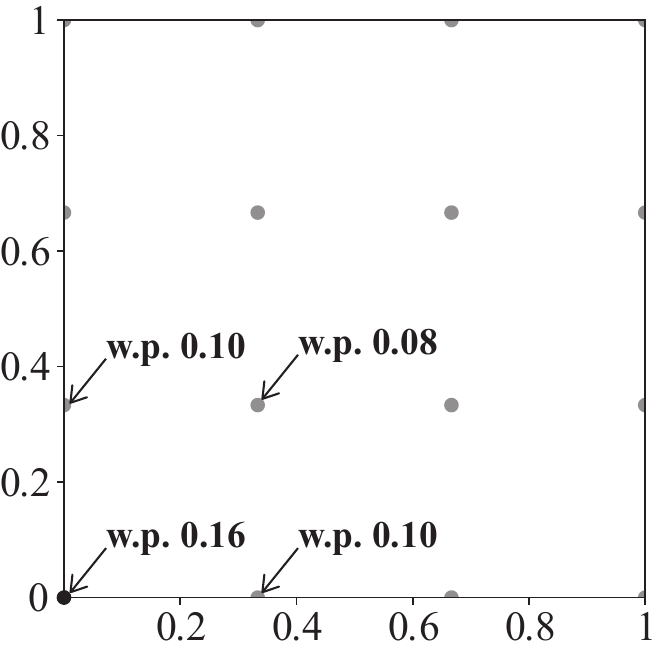}
        \caption{$m = 4 \times 4$ (gray points)}
        \label{fig:appendix_exp_mechanism_1}
    \end{subfigure}
    \hfill
    \begin{subfigure}[b]{0.46\linewidth}
        \centering
        \includegraphics[width=\textwidth]{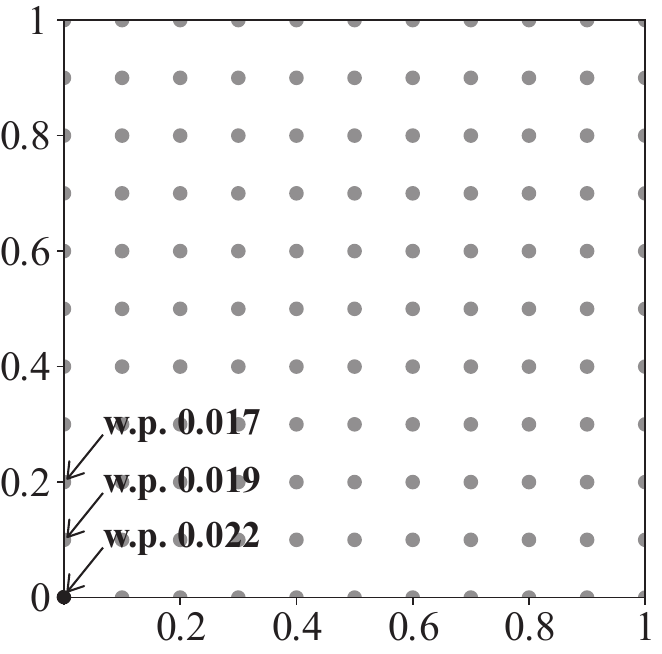}
        \caption{$m = 11 \times 11$ (gray points)}
        \label{fig:appendix_exp_mechanism_2}
    \end{subfigure}
    \caption{The Exponential mechanism $\mechanism_{\text{exp}}(0,0)$ with $\varepsilon=4$ on two discrete location spaces within $[0,1] \times [0,1]$.
        The number and arrangement of points affect the mechanism.}
    \label{fig:appendix_exp_mechanism}
\end{figure}

\begin{figure}
    \centering
    \begin{subfigure}[b]{0.46\linewidth}
        \centering
        \includegraphics[width=\textwidth]{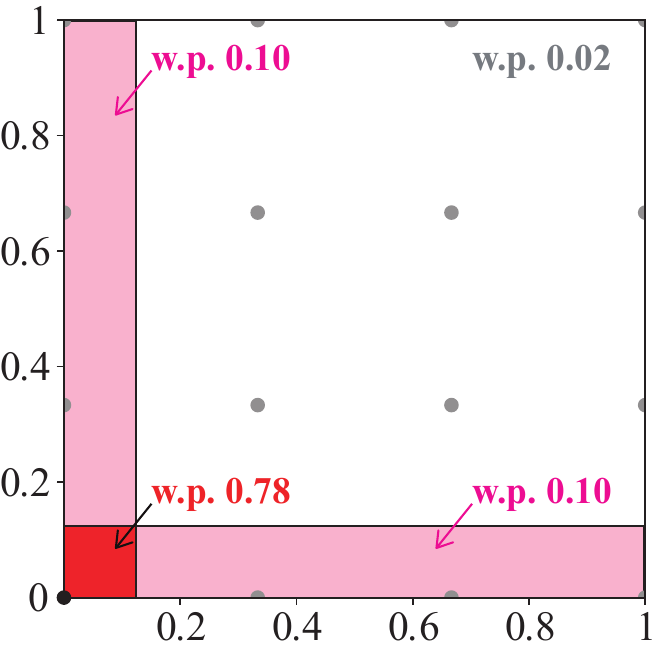}
        \caption{$\traspace = [0,1] \times [0,1]$}
    \end{subfigure}
    \hfill
    \begin{subfigure}[b]{0.46\linewidth}
        \centering
        \includegraphics[width=\textwidth]{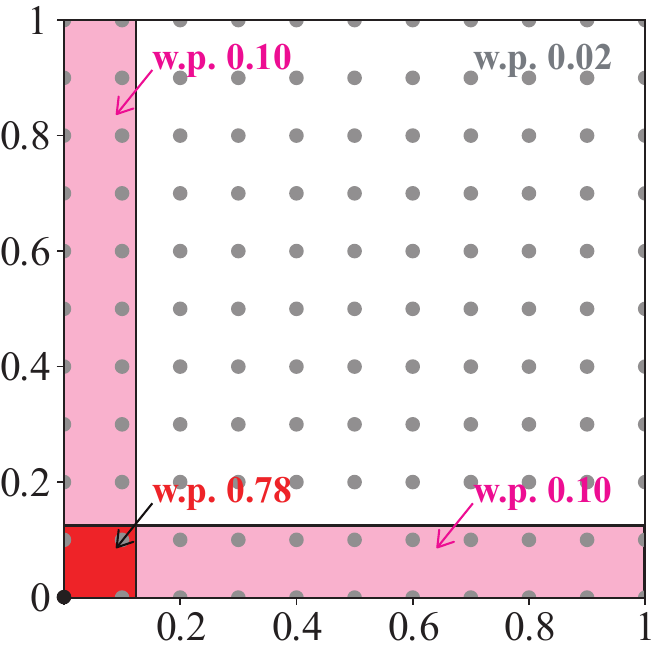}
        \caption{$\traspace = [0,1] \times [0,1]$}
    \end{subfigure}
    \caption{\oursc mechanism $\mechanism_{\text{\oursc}}(0,0)$ with $\varepsilon=4$ on $\traspace = [0,1] \times [0,1]$,
        which encompasses both discrete location spaces in Figure~\ref{fig:appendix_exp_mechanism}.
        The sampling probability is defined over areas rather than individual points.}
    \label{fig:appendix_tracsc}
\end{figure}

\subsubsection{Affected by the Number and Arrangement of Locations} \label{appendix:exp_mechanism_impact}

Because the sampling probability depends on the score function, the spatial distribution and relative distances between locations further affect the probability assigned to each location.

Figure~\ref{fig:appendix_exp_mechanism} illustrates the Exponential mechanism $\mechanism_{\text{exp}}$ applied to the sensitive location $\tau = (0,0)$ with $\varepsilon=4$ on two discrete location spaces within $[0,1] \times [0,1]$.
Specifically, Figure~\ref{fig:appendix_exp_mechanism_1} presents a $4 \times 4$ grid of locations, while Figure~\ref{fig:appendix_exp_mechanism_2} depicts an $11 \times 11$ grid.
We can observe that the sampling probabilities of $\mechanism_{\text{exp}}$ differ for the same sensitive location $\tau = (0,0)$,
even though both location spaces are contained within the same $[0,1] \times [0,1]$ area.
When the number of locations increases and the arrangement becomes denser (i.e. $m = 11 \times 11$),
the sampling probabilities for each location become smaller and more uniform.

In contrast, Figure~\ref{fig:appendix_tracsc} illustrates the \oursc mechanism $\mechanism_{\text{\oursc}}$ applied to the same sensitive location $\tau = (0,0)$ with $\varepsilon=4$ over the $[0,1] \times [0,1]$ area.
Unlike the Exponential mechanism, \oursc is designed for continuous spaces, where the sampling probability is defined over areas rather than individual points.

To apply \oursc to the discrete spaces shown in Figure~\ref{fig:appendix_exp_mechanism},
one can simply round each perturbed location to its nearest discrete point.
Since this rounding is a post-processing step, it does not affect the LDP guarantee.

\subsection{Limitations of Applying Discrete LDP Mechanisms to Continuous Spaces (Section~\ref{sec:introduction})} \label{appendix:discrete_to_continuous}

Although a continuous space can be discretized before applying discrete LDP mechanisms,
this approach inherits the privacy--utility--efficiency issues of discrete methods and introduces additional challenges in choosing an appropriate discretization strategy.

\textbf{Gridding the continuous space.} A common discretization approach is to partition the continuous space into uniform grid cells
and treat each cell as a discrete input and output of the discrete mechanism.
The reported cell can then be post-processed by randomly sampling a location within that cell.

\begin{figure}
    \centering
    \begin{subfigure}[b]{0.46\linewidth}
        \centering
        \includegraphics[width=\textwidth]{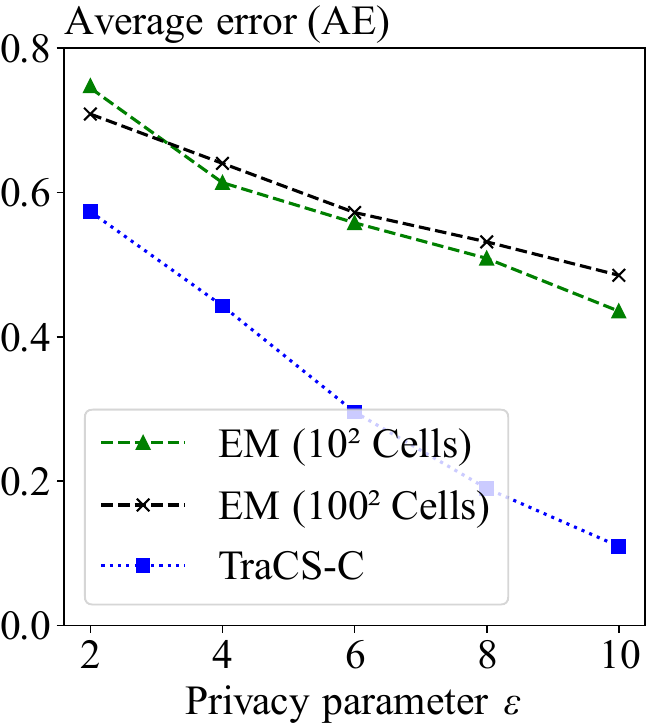}
        \caption{Error of the gridding approach with different cell sizes.}
    \end{subfigure}
    \hfill
    \begin{subfigure}[b]{0.48\linewidth}
        \centering
        \includegraphics[width=\textwidth]{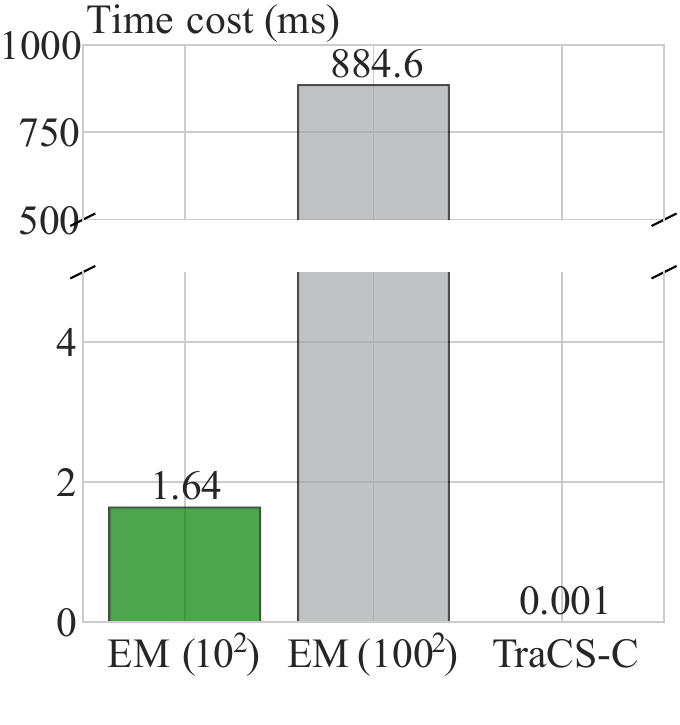}
        \caption{Average time cost of perturbing one location.}
    \end{subfigure}
    \caption{Performance of the gridding approach with different cell sizes in a continuous space.
    Coarse-grained grids (e.g. $10 \times 10$) perform better when $\varepsilon$ is small, whereas fine-grained grids (e.g. $100 \times 100$) perform better when $\varepsilon$ is large, at the cost of higher computational time.
    In contrast, \ours achieves better utility than both gridding baselines across all $\varepsilon$ values, with negligible time overhead.}
    \label{fig:appendix_griding}
\end{figure}

\textbf{Dilemma.}
However, the gridding approach faces a fundamental dilemma in choosing the grid cell size, which directly affects the privacy--utility--efficiency tradeoff.
(i) If the grid cells are too large (i.e. coarse-grained gridding), the mechanism is more likely to output the true cell, but the intra-cell error 
(i.e. the distance between the true location and the post-processed location sampled within the same cell) can be large.
(ii) If the grid cells are too small (i.e. fine-grained gridding), the intra-cell error can be reduced, but the probability of outputting the true cell decreases.
Meanwhile, the computational cost increases with the number of cells.

We empirically evaluate the gridding approach under different cell sizes in a continuous space and compare it with \oursc.
Specifically, we consider the $[0,1]\times[0,1]$ location space and two grid resolutions: $10\times 10$ and $100\times 100$.
For a fixed location, we generate $500$ perturbed locations using (i) the gridding approach with the Exponential mechanism (EM), 
where the score is defined by the distance between cell centers, and (ii) \oursc.
We then report the average error and the average time required to perturb a single location.
\ Figure~\ref{fig:appendix_griding} summarizes the results.
In the small-$\varepsilon$ regime, coarse-grained grids (e.g. $10 \times 10$) outperform fine-grained grids (e.g. $100 \times 100$),
whereas in the large-$\varepsilon$ regime, fine-grained grids perform better.
This trend is consistent with the above dilemma: it is difficult to select a single grid resolution that performs well across different $\varepsilon$ values.
In contrast, \oursc consistently achieves better utility than both gridding baselines for all tested $\varepsilon$ values, while incurring negligible time overhead.

\subsection{Dominant Sector Comparison Between the Strawman Approach and \oursd (Example~\ref{example:direction_perturbation})} \label{appendix:comparison_strawman}

\begin{figure}
    \centering
    \begin{subfigure}[b]{0.48\linewidth}
        \centering
        \includegraphics[width=\textwidth]{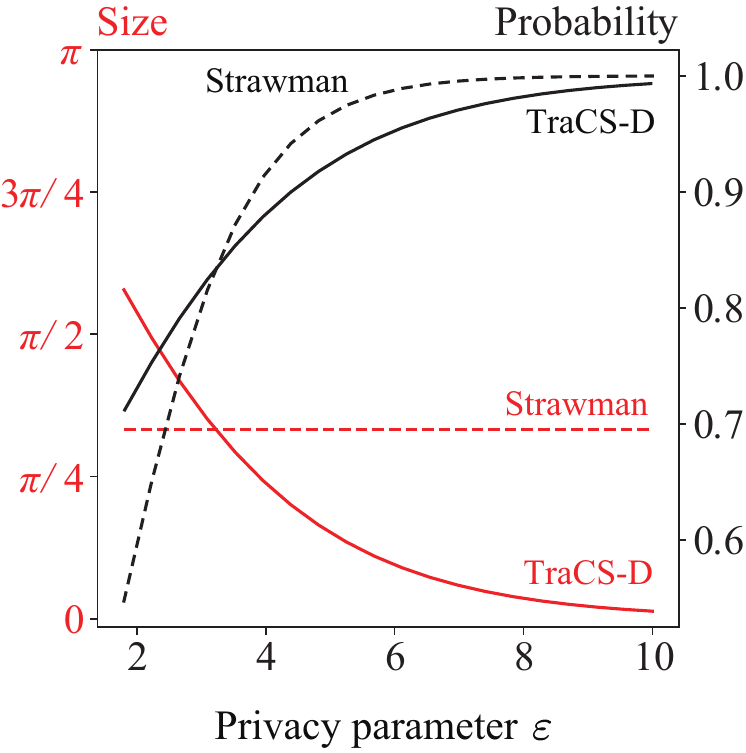}
        \caption{$k = 6$}
    \end{subfigure}
    \hfill
    \begin{subfigure}[b]{0.48\linewidth}
        \centering
        \includegraphics[width=\textwidth]{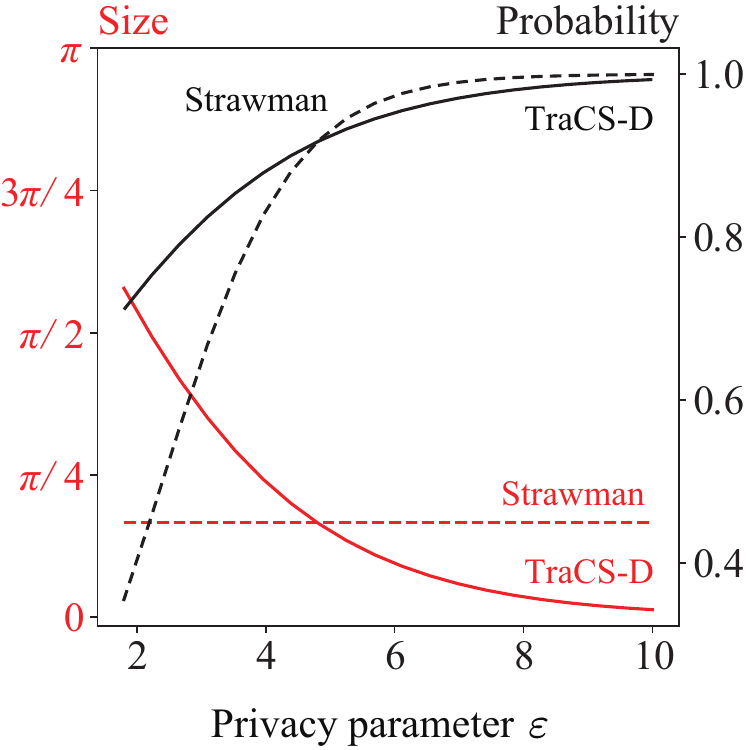}
        \caption{$k = 12$}
    \end{subfigure}
    \caption{Comparison of the size and probability of the dominant sector between \oursd and the strawman approach.
    \oursd (solid line) achieves a better trade-off between the size and probability of the dominant sector compared to the strawman approach (dashed line), 
    having better adaptiveness across different $\varepsilon$ regions.}
    \label{fig:strawman_vs_oursd}
\end{figure}

For direction perturbation, a small-size and high-probability dominant sector is desirable for more accurate perturbation.
However, these two properties are generally conflicting with each other due to the LDP constraint:
a smaller dominant sector typically has a lower probability.
\ We will show that \oursd achieves a better trade-off between the size and probability of the dominant sector 
compared to the strawman approach with any $k$ value.
Figure~\ref{fig:strawman_vs_oursd} shows two exemplary quantitative comparisons between them.

\textbf{Small $\varepsilon$ region.}
In this case, the dominant sector of \oursd (red solid line) is larger than the strawman approach, 
but the probability of the dominant sector (black solid line) is significantly higher,
especially when $k$ is large.
This indicates that, although the dominant sector of \oursd is larger in this region, it is sampled with a higher probability.

\textbf{Large $\varepsilon$ region.}
In such cases, the dominant sector of \oursd is significantly smaller than the strawman approach,
while maintaining almost the same probability as the strawman approach.

We can observe the effect of $k$ on the strawman approach.
Larger $k$ values in the strawman approach reduce the inner-sector error compared to \oursd, 
but also significantly decrease the probability of the dominant sector when $\varepsilon$ is small.
Conversely, smaller $k$ values increase the probability of the dominant sector when $\varepsilon$ is small,
but result in a large inner-sector error when $\varepsilon$ is large.
In contrast, \oursd adaptively balances these trade-offs by adjusting the dominant sector's size and probability according to $\varepsilon$.

\subsection{Detailed Form of $|\mathcal{D}_{r(\varphi)}|$ (Section~\ref{subsubsec:distance_perturbation})} \label{appendix:detailed_expression}

For each location $\tau_i$, denote its directions to the four endpoints of the rectangular
location space as $\varphi_1, \varphi_2, \varphi_3, \varphi_4$. We have:
\begin{equation*}
    \begin{split}
        \varphi_1 &= \mathrm{atan2} \left( b_{\text{end}} - b_i, a_{\text{end}} - a_i \right), \\
        \varphi_2 &= \mathrm{atan2} \left(b_{\text{end}} - b_i, a_{\text{sta}} - a_i \right), \\
        \varphi_3 &= \mathrm{atan2} \left(b_{\text{sta}} - b_i, a_{\text{sta}} - a_i \right) + 2\pi, \\
        \varphi_4 &= \mathrm{atan2} \left(b_{\text{sta}} - b_i, a_{\text{end}} - a_i \right) + 2\pi.
    \end{split}
\end{equation*}
Then, $|\mathcal{D}_{r(\varphi)}|$ can be derived using trigonometric functions,
which leads to the following four cases:
\begin{equation*}
    \begin{split}
        |\mathcal{D}_{r(\varphi)}| = 
        \begin{dcases}
            \frac{a_{\text{end}} - a_i}{\cos \varphi} & \text{if} \ \varphi \in [0, \varphi_1) \cup [\varphi_4, 2\pi), \\
            \frac{b_{\text{end}} - b_i}{\sin \varphi} & \text{if} \ \varphi \in [\varphi_1, \varphi_2), \\
            \frac{a_i - a_{\text{sta}}}{-\cos \varphi} & \text{if} \ \varphi \in [\varphi_2, \varphi_3), \\
            \frac{b_i - b_{\text{sta}}}{-\sin \varphi} & \text{if} \ \varphi \in [\varphi_3, \varphi_4),
        \end{dcases}
    \end{split}
\end{equation*}
depending on the direction $\varphi$ from reference location $\tau_i = (a_i, b_i)$ to the boundary of the rectangular location space
$[a_{\text{sta}}, a_{\text{end}}] \times [b_{\text{sta}}, b_{\text{end}}]$.

\subsection{Redesigned SW (Section~\ref{appendix:other_piecewise})} \label{appendix:resigned_sw}

We redesign the SW mechanism to achieve $\varepsilon$-LDP for the circular space $[0,2\pi)$
and the linear space $[0,1)$.
The key idea is to transform the original sampling distribution to these spaces while preserving the LDP constraint.
\ Similar to Definition~\ref{def:direction_perturbation}, the redesigned SW mechanism for direction perturbation
is defined as follows.

\begin{definition}[Redesigned SW for direction perturbation] \label{def:resigned_sw_direction}
    Given a sensitive direction $\varphi$ and a privacy parameter $\varepsilon$,
    redesigned SW for direction perturbation is a mechanism  $\mathcal{M}_\circ: [0,2\pi)\to [0,2\pi)$
    defined by:
    \begin{equation*}
        pdf[\mechanism_\circ(\varphi)=\varphi'] =
        \begin{dcases}
            p_{\varepsilon} & \text{if} \ \varphi' \in [l_{\varphi,\varepsilon}, r_{\varphi,\varepsilon}), \\
            p_{\varepsilon} / \exp{(\varepsilon)} & \text{otherwise},
        \end{dcases}
    \end{equation*}
    where $p_\varepsilon = \frac{1}{2\pi \varepsilon}(\exp(\varepsilon) - 1)$ is the sampling probability,
    and $[l_{\varphi,\varepsilon}, r_{\varphi,\varepsilon})$ is the sampling interval that
    \begin{equation*}
        \begin{split}
            l_{\varphi,\varepsilon} &= \left( \varphi - \pi\frac{\exp(\varepsilon)(\varepsilon - 1) + 1}{(\exp(\varepsilon) - 1)^2} \right) \mod 2\pi, \\
            r_{\varphi,\varepsilon} &= \left( \varphi + \pi\frac{\exp(\varepsilon)(\varepsilon - 1) + 1}{(\exp(\varepsilon) - 1)^2} \right) \mod 2\pi.
        \end{split}
    \end{equation*}
\end{definition}

Similar to Definition~\ref{def:distance_perturbation}, the redesigned SW mechanism for distance perturbation
is defined as follows.

\begin{definition}[Redesigned SW for distance perturbation] \label{def:resigned_sw_distance}
    Given a sensitive distance $\overline{r}(\varphi)$ and a privacy parameter $\varepsilon$,
    redesigned SW for distance perturbation is a mechanism $\mathcal{M}_{-}: [0,1)\to [0,1)$ that
    \begin{equation*}
        pdf[\mechanism_{-}(\overline{r}(\varphi))=\overline{r}'(\varphi)] =
        \begin{dcases}
            p_{\varepsilon} & \text{if} \ \overline{r}'(\varphi) \in [u, v), \\
            p_{\varepsilon} / \exp{(\varepsilon)} & \text{otherwise},
        \end{dcases}
    \end{equation*}
    where $p_{\varepsilon} = (\exp(\varepsilon) - 1) / \varepsilon$ is the sampling probability,
    and $[u,v)$ is the sampling interval that
    \begin{equation*}
        \begin{split}
            [u,v) &=
            \begin{cases}
                \overline{r}(\varphi) + [-C, C) & \text{if} \ \overline{r}'(\varphi) \in [C, 1-C), \cr
                [0, 2C) & \text{if} \ \overline{r}'(\varphi) \in [0, C), \cr
                [1-2C, 1) & \text{otherwise},
            \end{cases}
        \end{split}
    \end{equation*}
    with $C=(\exp(\varepsilon)(\varepsilon - 1)+ 1) / (2(\exp(\varepsilon) - 1)^2)$.
\end{definition}

Compared with Definition~\ref{def:direction_perturbation} and Definition~\ref{def:distance_perturbation},
the redesigned SW mechanisms have a larger $p_\varepsilon$ and a narrower $[l, r)$ or $[u, v]$ interval for sampling.
We can treat the redesigned SW mechanisms as a more aggressive perturbation mechanism:
it tries to sample from a narrow dominant sector with a higher probability ($p_\varepsilon$), 
which also leads to a higher probability ($p_\varepsilon/\exp(\varepsilon)$) of non-dominant sectors due to the LDP constraint.

\subsection{Experimental Results for Extensions (Section~\ref{appendix:other_piecewise})} \label{appendix:resigned_sw_exp}

\subsubsection{Effect of Different Piecewise-based Mechanisms} \label{sec:exp:redesigned_sw}
Figure~\ref{fig:exp:3} compares the performance of different piecewise-based mechanisms for \oursd.
In addition to the mechanisms defined in Definition~\ref{def:direction_perturbation} and Definition~\ref{def:distance_perturbation},
we name the redesigned SW for \oursd as \oursd (R-SW),
as detailed in Appendix~\ref{appendix:resigned_sw}.
\ We can observe that \oursd consistently exhibits smaller errors than \oursd (R-SW) across all the $\varepsilon$ values.
Statistically, the mean AE of \oursd is $52.6\%$ of \oursd (R-SW) across all the $\varepsilon$ values.
This difference comes from the MSE of their sampling distributions, 
where Definition~\ref{def:direction_perturbation} and Definition~\ref{def:distance_perturbation}
have smaller MSE than R-SW.

\subsubsection{\oursd for Circular Area}

Figure~\ref{fig:exp:4} illustrates an example of \oursd applied to a circular area $\mathcal{S} = [0,2\pi) \otimes [0,1)$.
We set the sensitive location $\tau$ as $(\pi, 0.5)$ and the privacy parameter $\varepsilon = 5$,
then collect $50$ random samples of perturbed locations using \oursd.
In this case, the dominant area is $[0.87\pi, 1.13\pi)\otimes [0.32, 0.67)$, as determined by the perturbation mechanisms
in Definition~\ref{def:direction_perturbation} and Definition~\ref{def:distance_perturbation}.
We can observe that the perturbed locations are concentrated in the dominant area.

\begin{figure}[t]
    \begin{minipage}{0.45\linewidth}
        \centering
        \includegraphics[width=0.98\textwidth]{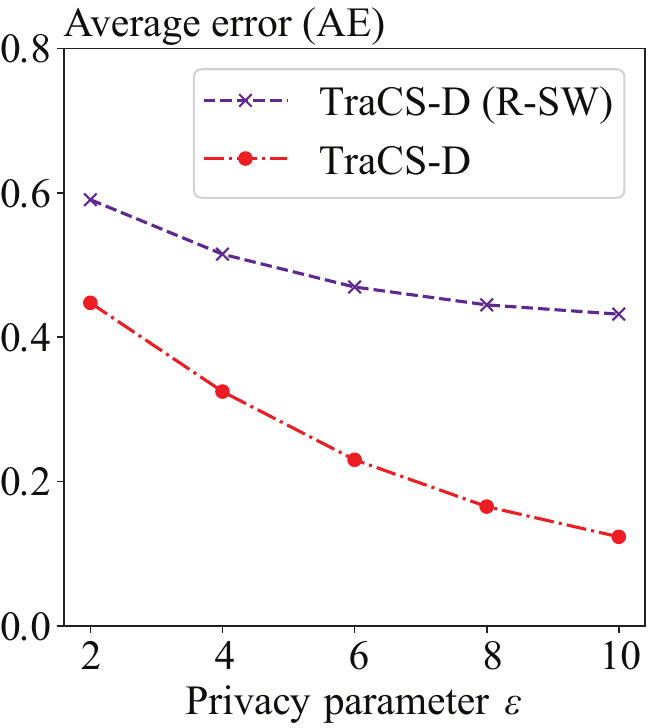}
        \caption{Different piecewise-based mechanisms.}
        \label{fig:exp:3}
    \end{minipage}
    \hfill
    \begin{minipage}{0.45\linewidth}
        \centering
        \includegraphics[width=0.98\textwidth]{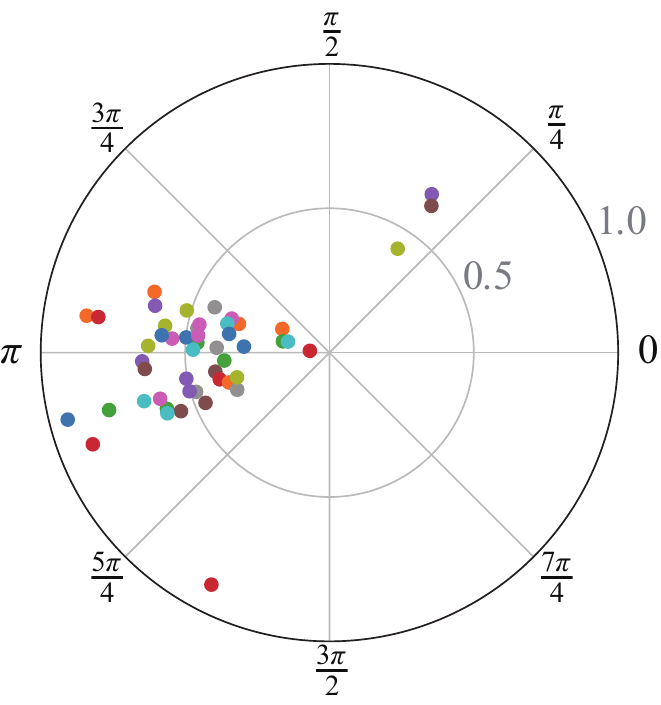}
        \vskip 0.5em
        \caption{Samples of \oursd at $\tau = (\pi, 0.5)$ for circular area.}
        \label{fig:exp:4}
    \end{minipage}
\end{figure}

\subsection{2D Laplace Mechanism with Truncation (Section~\ref{sec:continuous_exp_setup})} \label{appendix:2d_laplace}

For 2D location spaces equipped with the Euclidean distance, a standard Laplace-based approach is the Planar Laplace mechanism~\cite{DBLP:conf/ccs/AndresBCP13}.\footnote{
        We omit the standard (truncated) Laplace mechanism for $\mathbb{R}$, as it is designed for $L_1$ distance and has been shown to have worse data utility compared to piecewise-based mechanisms~\cite{DBLP:journals/popets/ZhengMH25}.
        }
It was originally proposed for \emph{geo-indistinguishability}, and it also provides an LDP guarantee with a calibrated sensitivity.

\begin{definition}[Planar Laplace mechanism]
    Let $\mathcal{X} \subseteq \mathbb{R}^2$ be the input domain. The Planar Laplace mechanism
    $\mechanism_{\mathrm{PL}}: \mathcal{X} \to \mathbb{R}^2$ is defined by
    \[
        \mechanism_{\mathrm{PL}}(x) = x + \eta,
    \]
    where $\eta \in \mathbb{R}^2$ is a noise vector with density
    \[
        pdf[\eta] = \frac{\varepsilon^2}{2\pi}\exp\bigl(-\varepsilon\|\eta\|_2\bigr).
    \]
\end{definition}

\textbf{Sensitivity calibration for an LDP guarantee.}
By the triangle inequality, for any $x_1,x_2 \in \mathcal{X}$ and any $y \in \mathbb{R}^2$, we have
\[
    \frac{pdf[y-x_1]}{pdf[y-x_2]}
    = \exp\Bigl(\varepsilon\bigl(\|y-x_2\|_2-\|y-x_1\|_2\bigr)\Bigr)
    \le \exp\bigl(\varepsilon\|x_1-x_2\|_2\bigr).
\]
Therefore, over the domain $\mathcal{X}$, the mechanism satisfies $(\varepsilon\cdot \mathrm{diam}(\mathcal{X}))$-LDP, where
$\mathrm{diam}(\mathcal{X})$ is the diameter of $\mathcal{X}$.
Equivalently, to ensure $\varepsilon$-LDP for all pairs of inputs in $\mathcal{X}$, one can run the Planar Laplace mechanism with parameter
$\varepsilon/\mathrm{diam}(\mathcal{X})$.

\textbf{Truncation for bounded location spaces.}
The Planar Laplace mechanism outputs perturbed locations in $\mathbb{R}^2$, which may fall outside the bounded location space $\mathcal{S}$.
A common practice is to truncate the output to $\mathcal{S}$, e.g. by projecting it to the nearest point in $\mathcal{S}$.
This truncation is a post-processing step and does not affect the LDP guarantee.

\subsection{Privacy Parameter Assignment in \oursd (Section~\ref{sec:continuous_exp_setup})} \label{appendix:epsilon_d}

In \oursd, the overall perturbation error is jointly determined by the direction perturbation and the distance perturbation.
These two components are controlled by $\varepsilon_d$ for $\mechanism_\circ$ and $\varepsilon-\varepsilon_d$ for $\mechanism_{-}$, respectively.
As a result, the (theoretical) error bound of \oursd depends on both $\varepsilon$ and the privacy split $\varepsilon_d$.
In principle, the best choice of $\varepsilon_d$ is the one that minimizes this bound.

\begin{wrapfigure}{r}{0.34\linewidth}
    \centering
    \includegraphics[width=\linewidth]{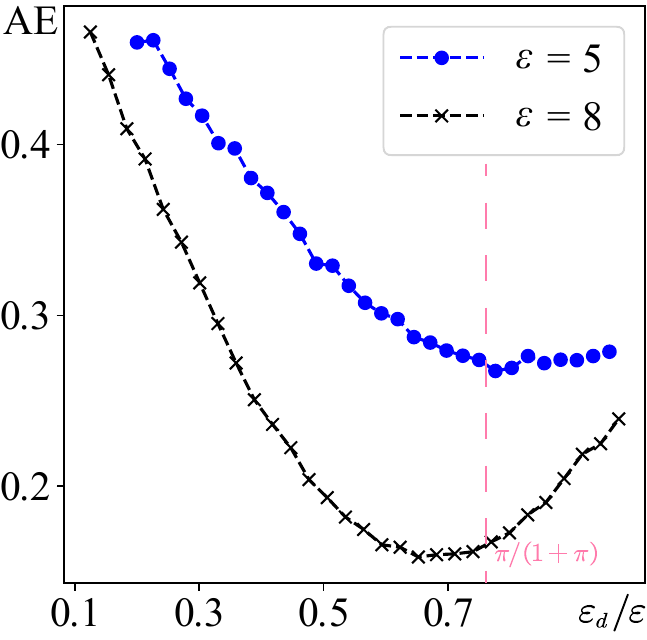}
    \label{fig:appendix:epsilon_d}
\end{wrapfigure}
However, the optimal split depends on $\varepsilon$ and on the shape of $\traspace$, which makes a universal closed-form choice impractical.
Using the same experimental setup as in the Figure~\ref{fig:exp:1_1}, we empirically evaluate how $\varepsilon_d$ affects the error of \oursd; the results are shown in the figure on the right.
We observe a consistent pattern: for each fixed $\varepsilon$, the error decreases initially and then increases as $\varepsilon_d$ varies from $0$ to the full budget~$\varepsilon$.

Across all tested $\varepsilon$ values, the minimizer $\varepsilon_d/\varepsilon$ is close to our heuristic split $\varepsilon_d=\pi/(1+\pi)$ (dashed pink line), which supports the use of this heuristic in practice.

\subsection{Privacy Parameter Assignment for a Whole Trajectory (Section~\ref{subsubsec:trajectory_level_epsilon})} \label{appendix:trajectory_level_epsilon}

\begin{figure}[t]
    \centering
    \begin{subfigure}[b]{0.465\linewidth}
        \centering
        \includegraphics[width=0.98\textwidth]{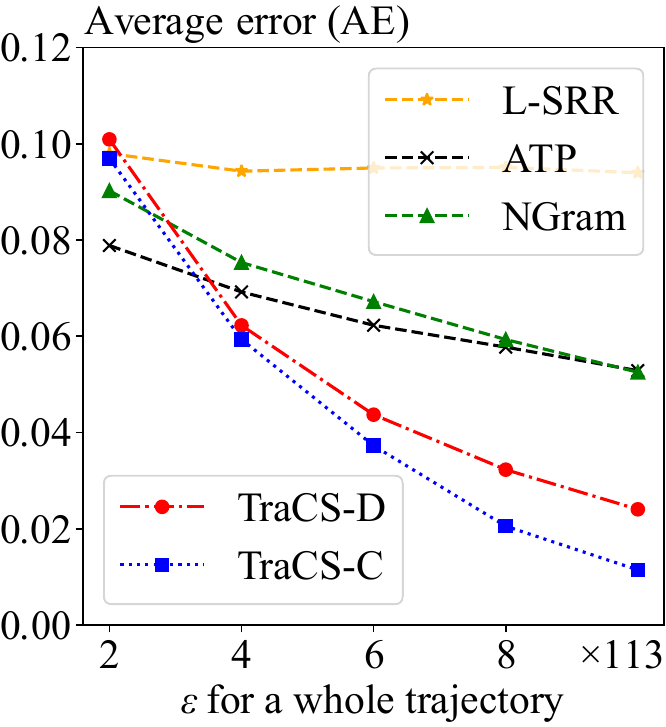}
        \caption{TKY dataset}
        \label{appendix:fig:trajectory_epsilon_tky}
    \end{subfigure}
    \hfill
    \begin{subfigure}[b]{0.465\linewidth}
        \centering
        \includegraphics[width=0.98\textwidth]{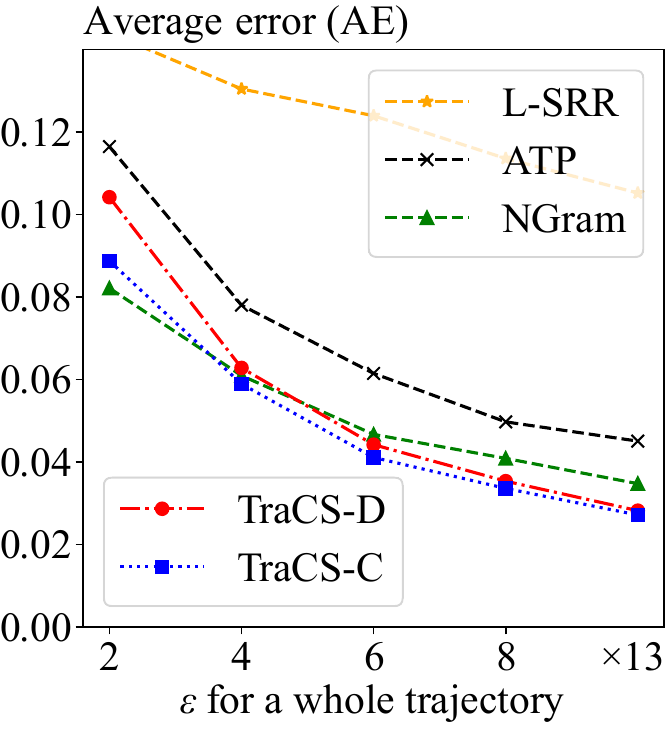}
        \caption{CHI dataset}
        \label{appendix:fig:trajectory_epsilon_chi}
    \end{subfigure}
    \caption{Comparison with trajectory-level $\varepsilon$ assignment. 
    The average trajectory length is $113$ for TKY and $13$ for CHI. 
    \ours's AE decreases fast as $\varepsilon$ increases, outperforming other discrete-space mechanisms when $\varepsilon$ is large.}
    \label{appendix:fig:trajectory_epsilon}
\end{figure}

\begin{figure}[t]
    \centering
    \begin{subfigure}[b]{0.465\linewidth}
        \centering
        \includegraphics[width=0.98\textwidth]{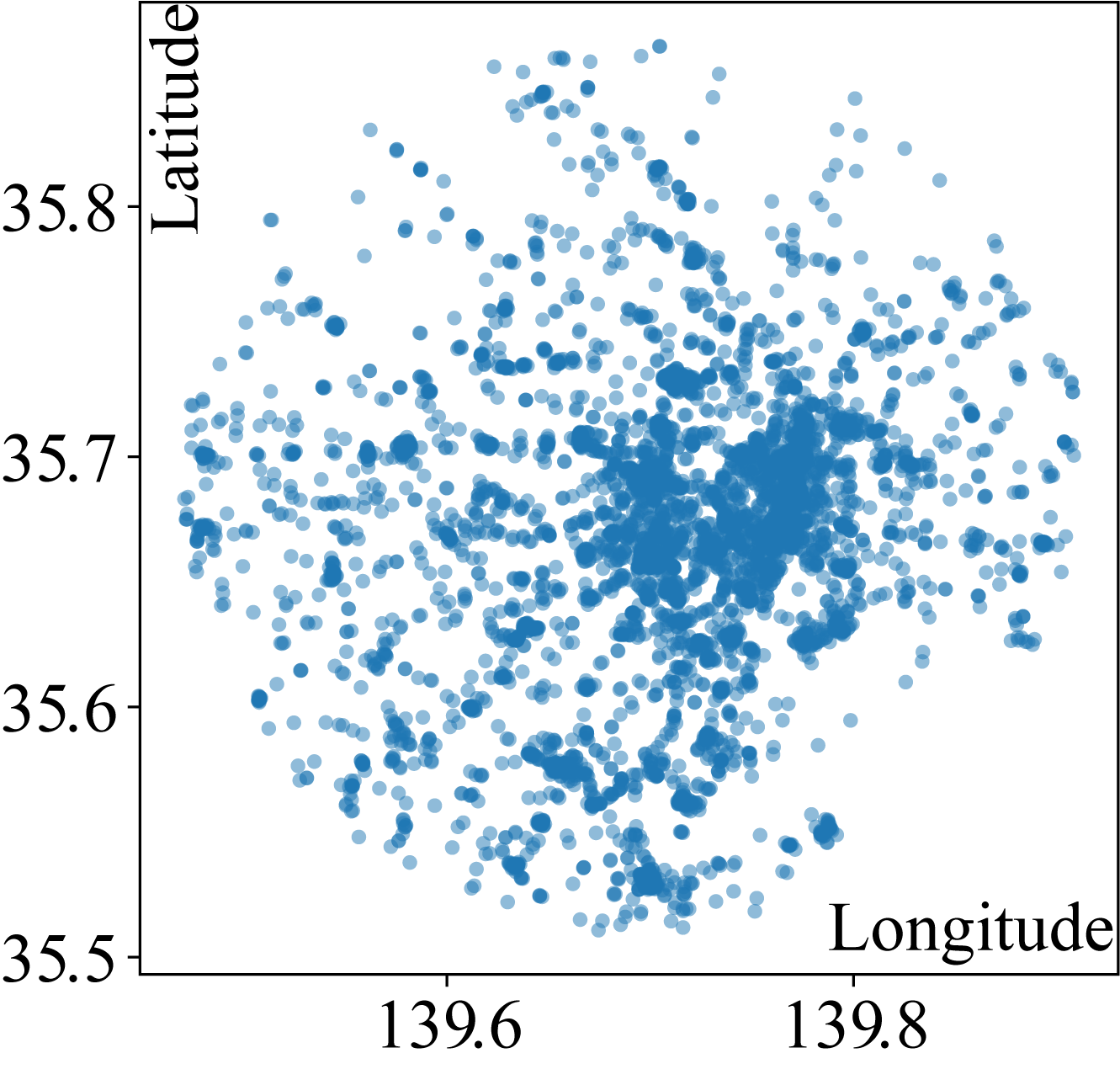}
        \caption{TKY location space}
        \label{appendix:fig:tky_location_space}
    \end{subfigure}
    \hfill
    \begin{subfigure}[b]{0.465\linewidth}
        \centering
        \includegraphics[width=0.98\textwidth]{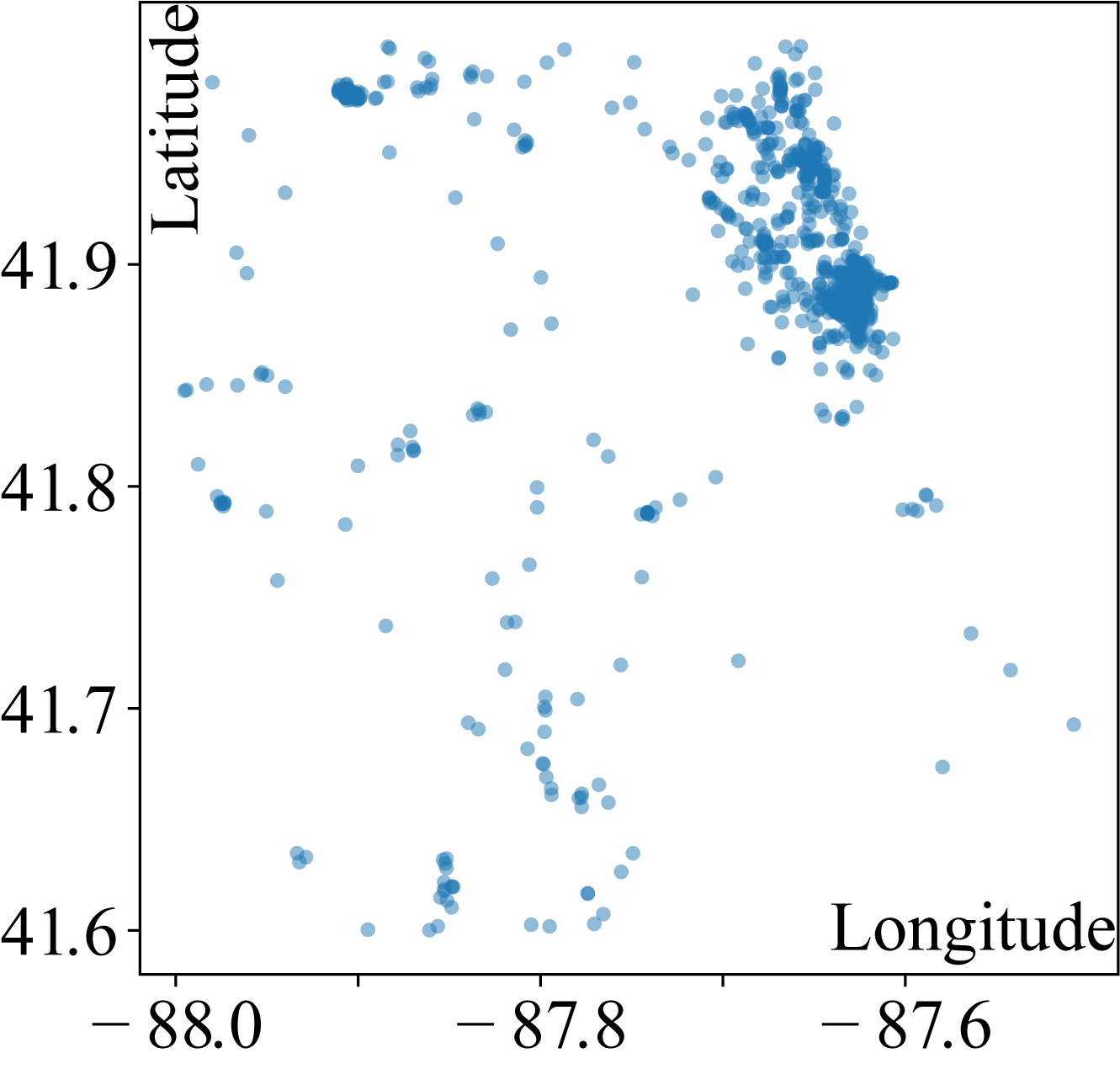}
        \caption{CHI location space}
        \label{appendix:fig:chi_location_space}
    \end{subfigure}
    \caption{Location spaces for TKY and CHI datasets.
    \ours operates over the whole rectangular area that encloses the discrete location space of each city, 
    while other discrete-space mechanisms operate only on the discrete locations (blue points).}
    \label{appendix:fig:location_spaces}
\end{figure}

We compare \ours with NGram, L-SRR, and ATP under a trajectory-level privacy parameter assignment on the TKY and CHI datasets.
The average trajectory lengths are $113$ (TKY) and $13$ (CHI). Accordingly, we assign a trajectory-level budget of $\varepsilon \times 113$ for TKY and $\varepsilon \times 13$ for CHI.
This scaling makes the trajectory-level setting more comparable to the location-level $\varepsilon$ assignment in Figure~\ref{fig:exp:real_2}.
The results are shown in Figure~\ref{appendix:fig:trajectory_epsilon}.

We observe that \ours has larger AEs than the discrete-space mechanisms when $\varepsilon$ is small, but its errors decrease rapidly as $\varepsilon$ increases.
As a result, \ours outperforms the discrete-space mechanisms in the large-$\varepsilon$ regime.
Compared with assigning $\varepsilon$ at the location level in Figure~\ref{fig:exp:real_2}, the overall trends remain similar across all mechanisms:
\ours starts to outperform all other discrete-space mechanisms when $\varepsilon \approx 4$ per location on both datasets.
Among the discrete-space mechanisms, ATP achieves the lowest AEs on TKY, whereas NGram (i.e. the Exponential mechanism with a reachable set) attains the lowest AEs on CHI, 
due to CHI's smaller location space than TKY.

As discussed in Section~\ref{sec:real_discrete}, the larger average errors (AEs) of \ours in the small-$\varepsilon$ regime can be attributed to its enlarged effective location space.
\ours is designed for continuous spaces and perturbs locations over the rectangular area that encloses the discrete location space of each city, as illustrated in Figure~\ref{appendix:fig:location_spaces}.
The discrete location set in TKY is comparatively dense and evenly distributed, whereas the location set in \mbox{CHI} is much sparser.
Consequently, the gap between the discrete set and its enclosing rectangle is larger for CHI, 
which amplifies the utility loss of \ours and explains its weaker performance on CHI (relative to discrete-space mechanisms) compared with TKY.